\documentclass[lettersize,journal]{IEEEtran}
\usepackage{amsmath,amsfonts,amssymb} 

\usepackage{graphicx}
\usepackage{array}
\usepackage[caption=false,font=normalsize,labelfont=sf,textfont=sf]{subfig}
\usepackage{float}
\usepackage{placeins}
\usepackage{refcount}
\usepackage{dblfloatfix}

\usepackage{algorithm}
\usepackage{algorithmic}

\usepackage{tikz}
\usetikzlibrary{arrows.meta}
\usepackage{pgfplots}
\pgfplotsset{compat=1.18}

\usepackage{textcomp}
\usepackage{url}
\usepackage{cite}
\usepackage{verbatim}

\newcommand{\uprightref}[1]{\textup{\ref{#1}}}

\newtheorem{theorem}{Theorem}
\newtheorem{definition}{Definition}
\newtheorem{proposition}{Proposition}
\newtheorem{lemma}{Lemma}

\newenvironment{proof}
{\par\indent\textit{Proof:}\enspace}
{\hfill$\blacksquare$\par}

\newcounter{optprob}
\renewcommand{\theoptprob}{P\arabic{optprob}}

\makeatletter
\newcommand{\optprob}[1]{%
  \refstepcounter{optprob}%
  \ltx@label{#1}%
  \textnormal{(\theoptprob)}%
}
\makeatother

\newcommand{\probref}[1]{\textnormal{(\ref{#1})}}

\begin{document}
\title{Cross-Layer Analysis and Optimization for Movable Antenna Systems {with Delay-Outage Constraints}}
\vspace{-0.15cm}
\author{Haitao Yu, 
Weidong Mei,~\IEEEmembership{Member,~IEEE}, 
Lei Qian,
and Rui Zhang,~\IEEEmembership{Fellow,~IEEE}
\vspace{-0.5cm}
\thanks{Haitao Yu and Weidong Mei are with the National Key Laboratory of Wireless Communications, University of Electronic Science and Technology of China, Chengdu 610000, China (email:HaitaoYu0803@163.com, wmei@uestc.edu.cn). \textit{(Corresponding author: Weidong Mei.)}}
\thanks{Lei Qian is with the Tianjin Key Laboratory of Optoelectronic Detection Technology and System, Tiangong University, Tianjin, 300160, China (email: qianlei@tiangong.edu.cn).}
\thanks{Rui Zhang is with the Department of Electrical and Computer Engineering, National University of Singapore, Singapore 117583 (e-mail: elezhang@nus.edu.sg).}}
\markboth{Journal of \LaTeX\ Class Files,~Vol.~14, No.~8, August~2021}%
{Shell \MakeLowercase{\textit{et al.}}: A Sample Article Using IEEEtran.cls for IEEE Journals}


\renewcommand{\IEEEtitletopspaceextra}{-0.3cm}
\IEEEaftertitletext{\vspace{-0.05cm}}
\maketitle

\begin{abstract}
Movable antennas (MAs) have emerged as a promising technology for enhancing wireless communication performance. However, in delay-sensitive systems, MA movement incurs non-negligible delay, which may degrade the service process for bursty traffic. This paper investigates a delay-outage problem in an MA-enhanced multiuser multiple-input multiple-output (MU-MIMO) downlink system. Unlike existing studies that mainly focus on physical-layer performance optimization, a cross-layer framework is developed that accounts for the overall delay due to antenna movement, along with service interruptions and queue accumulation at the medium access control (MAC) layer. Considering the different timescales of antenna movement and packet-level queue evolution, we establish a two-timescale transmission protocol and a cross-layer queueing model with Markov-modulated bursty arrivals. A long-term transmit-power minimization problem is then formulated to jointly optimize the destination antenna positions and transmit beamforming under delay-outage and practical MA movement constraints. To tackle the probabilistic delay-outage constraints, a backlog martingale is constructed to derive a deterministic sufficient condition in terms of service rate, thereby transforming the probabilistic constraints into deterministic service rate requirements. An alternating optimization (AO) algorithm based on a short-horizon approximation is further developed to obtain a high-quality suboptimal solution. Numerical results demonstrate that the proposed algorithm outperforms baseline schemes by effectively balancing the MAs' spatial reconfiguration capability and the temporal service loss.
\end{abstract}
\vspace{-0.15cm}
\begin{IEEEkeywords}
Movable antennas, movement delay, martingale theory, delay-outage probability, cross-layer design.  
\vspace{-0.25cm}
\end{IEEEkeywords}

\section{Introduction}
\label{sec:introduction}
As sixth-generation (6G) mobile communication systems evolve to support emerging applications such as immersive communications, smart manufacturing, wireless networks are expected to provide high data rates, massive connectivity, and energy-efficient transmission while ensuring reliable low-latency services. For delay-sensitive applications, excessive latency can destabilize control loops, cause task failures, and create safety risks \cite{pang2025communication,khattak2023impact}. Since stringent statistical delay guarantees typically require additional resources and thereby reduce transmission efficiency, how to balance efficiency and statistical delay performance under limited spectrum, power, and hardware resources remains a fundamental challenge. 

On the other hand, multi-antenna transmission is fundamental to improving
wireless capacity, coverage, and reliability. Conventional fixed-position antenna (FPA) systems achieve spatial diversity gains mainly through larger arrays and beamforming optimization, but their degrees of freedom are constrained by fixed array geometries. Although adding antennas and/or radio-frequency (RF) chains may help boost the spatial diversity, this comes at the expense of dramatically higher hardware cost and processing complexity particularly for ultra-high frequency bands. Moreover, this barely improves beam-steering capability, as the spatial correlation among steering vectors corresponding to different directions remains fixed. In this context, movable antenna (MA) technology offers an alternative by repositioning antenna elements within a confined region. By this means, MA systems can exploit favorable locations to enhance desired signals, suppress interference, and improve beamforming flexibility \cite{zhu2026tutorial}.

Motivated by these advantages, extensive studies have investigated MA performance optimization under various system setups. For example, based on the field-response channel model developed in \cite{zhu2024modeling}, the authors in \cite{ma2024mimo} characterized the capacity of a point-to-point MA multiple-input multiple-output (MIMO) system, while the authors in \cite{zhu2024multiuser,xiao2024base,wu2023discrete,wang2024interference,yang2024flexible} studied rate maximization problems in MA-enhanced multi-user communication systems. Later studies have also investigated MA position optimization for non-orthogonal multiple access (NOMA) systems \cite{li2024noma}, multicasting systems \cite{gao2024multicast}, wideband systems \cite{zhu2024wideband}, spectrum sharing and sensing \cite{wei2024spectrum}\cite{ieee11478080}, physical-layer security \cite{hu2024secure,shen2025movable}, integrated sensing and communication \cite{lyu2025isac,jiang2025movable}, near-field communications \cite{zhu2025nearfield}, near-field sensing \cite{wang2025nearfieldsensing}, and short-packet transmission \cite{he2024shortpacket}. A more general six-dimensional MA (6DMA) model with joint antenna position and orientation control was further studied in \cite{shao2025_6dma_distribution,shao2025_6dma_discrete,shao2025_6dma_opportunities}. Unlike the above studies that optimize MA positions in a continuous domain based on the field-response channel, the authors in \cite{mei2024graph} and \cite{mei2024security_discrete} proposed discretizing the antenna movement region into a set of discrete sampling points. This approach converts the continuous antenna position optimization into a discrete point selection problem and eliminates the need for a continuous channel model. Although these studies confirm the benefits of MAs for channel reconfiguration, they primarily focus on the system performance achievable at the optimized MA positions, without incorporating the effects of MA movement delay.

However, in practical MA systems, antenna movement generally incurs non-negligible delay, consuming time resources that would otherwise be available for information transmission. To address this issue, the authors in \cite{li2026trajectory} investigated trajectory optimization for MAs to minimize the delay required to reposition antennas from their initial positions to the {optimized target positions}. Furthermore, the authors in \cite{wang2026throughput} incorporated antenna movement delay into communication performance analysis and revealed a fundamental tradeoff between movement delay and communication performance.

Despite these advances, existing studies on MAs have mainly focused on physical-layer performance, without examining their impact on queue evolution and delay quality of service (QoS) at the medium access control (MAC) layer. In short-packet systems such as industrial Internet of Things (IIoT) networks, bursty arrivals and service interruptions during MA movement can cause substantial queue accumulation. As a result, the existing backlog remains unserved while new arrivals further increase the queue length. Such queue buildup prolongs packet waiting times and deteriorates the delay-QoS performance. Therefore, MA movement creates a cross-layer coupling between physical-layer position control and MAC-layer delay-outage performance that cannot be fully characterized by the existing physical-layer-oriented studies.

This gap calls for a cross-layer analysis that connects MA movement with packet-level delay performance. It is worth noting that delay-outage performance analysis and delay-aware resource allocation have been widely studied at the MAC layer. Effective bandwidth (EB) and effective capacity (EC) are classical large-deviations-based tools that characterize the asymptotic decay of steady-state backlog and delay tails through the logarithmic moment-generating functions of the arrival and service processes, respectively \cite{elwalid1993effective,wu2003effective}. These approaches have been widely applied to delay-sensitive resource allocation in vehicular communications \cite{guo2019vehicular}, NOMA systems \cite{zeng2023noma}, visible light networks \cite{qian2021secure}, and industrial wireless networks \cite{li2024crosslayer}. However, their asymptotic approximations may be inaccurate at finite delay thresholds. Martingale-based approaches instead derive backlog and delay-outage bounds by analyzing sample-path threshold-crossing events using maximal inequalities and optional-stopping arguments \cite{poloczek2014envelope,poloczek2015service}. Such methods have been applied to vehicular networks \cite{hu2017vanet}, tactile communications \cite{yu2020tactile}, and heterogeneous satellite networks \cite{zhang2025satellite}.  More recently, sliding-block martingale analysis was developed to capture delay behavior over multiple observation scales with lower martingale-parameter determination complexity \cite{dang2026sliding}. Compared with EB/EC approximations, martingale methods can provide non-asymptotic delay-outage bounds under suitable conditions.

Nevertheless, the above studies mainly consider FPA systems, where the service process is typically characterized over a long observation horizon without explicitly accounting for interruptions induced by antenna movement. To the best of the authors’ knowledge, {cross-layer designs} for MA systems have not yet been investigated in prior work.

Motivated by the above, this paper investigates a delay-outage probability constrained design problem in the downlink of an MA-enhanced multiuser multiple-input {multiple}-output (MU-{M}IMO) system. We characterize the cross-layer impact of MA movement delay and optimize MA positions and transmit beamforming under prescribed delay-outage probability constraints  which account for {the overall delay due to antenna movement}, along with {service interruptions and queue accumulation at the MAC layer}. The main contributions are summarized as follows.

\vspace{-0.05cm}
\begin{itemize}

\item Due to the different timescales of antenna position control and packet-level queue evolution, a two-timescale transmission protocol is developed. Furthermore, by modeling bursty arrivals as a Markov-modulated process, a cross-layer queueing model is established to characterize the coupling among MA movement, service interruption, traffic accumulation, and post-movement service capability.

\item To enable power-efficient MA-enhanced transmission for delay-sensitive traffic, a long-term average transmit-power minimization problem is formulated based on the two-timescale protocol and cross-layer queueing model. The destination positions of the base station (BS)-side transmitter (Tx)-MAs and user-side receiver (Rx)-MAs, as well as the transmit beamforming vectors, are jointly optimized subject to prescribed delay-outage probability constraints and practical MA movement constraints.

\item To tackle the probabilistic delay-outage constraints, a backlog martingale is constructed to jointly characterize the inherited queue backlog, arrivals accumulated during MA movement, and the post-movement service process. By this means, a deterministic sufficient condition for the delay-outage constraints is derived with respect to the service rate, which renders the long-term power minimization problem more tractable.

\item Last, an alternating optimization (AO) algorithm based on a short-horizon approximation is developed to obtain a high-quality suboptimal solution to the considered problem, by combining a two-layer problem reformulation with successive convex approximation (SCA). Numerical results demonstrate that the proposed algorithm outperforms the considered baseline schemes by effectively satisfying the prescribed delay-outage constraints and balancing the channel-reshaping capability of MAs against the service-interruption cost.
\end{itemize}
\vspace{-0.1cm}

\begin{figure*}[!t]
    \centering
    \vspace{-0.8cm}
    \includegraphics[width=17cm]{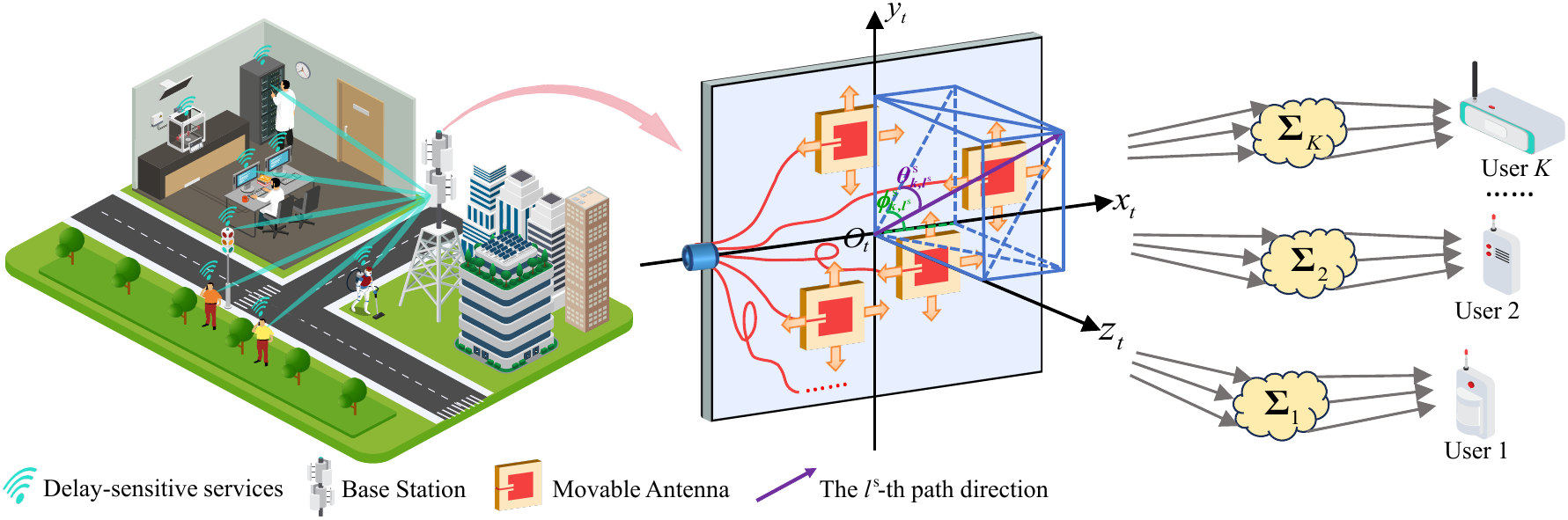}
    \vspace{-0.3cm}
    \caption{Illustration of the MA-enhanced MU-MIMO downlink system for delay-sensitive services.}
    \vspace{-0.3cm}
    \label{fig:system_model}
\end{figure*}

The remainder of this paper is organized as follows. {Section~II presents the system model, the two-timescale transmission protocol, and the long-term power-minimization problem. Section~III develops the backlog-martingale-based delay analysis and derives the sufficient conditions for the probabilistic delay-outage constraints.} Section~IV presents the proposed AO-based algorithm. Section~V presents numerical results, and Section~VI concludes the paper.

{\it{Notations:}}
Italic letters denote scalars, while boldface and calligraphic letters denote {vectors/matrices} and sets{, respectively}. $\mathbb{R}$ and $\mathbb{C}$ denote the {sets of} real and complex {numbers, respectively}. $(\cdot)^T$ and $(\cdot)^H$ denote transpose and conjugate transpose{, respectively}. {The notation $|\cdot|$ denotes the modulus of a scalar or the cardinality of a set, depending on the context, while} $\|\cdot\|_2$, $\mathrm{Tr}(\cdot)$, and $\det(\cdot)$ denote the Euclidean norm, trace, and determinant{, respectively}. $\Pr\{\cdot\}$ and $\mathbb{E}\{\cdot\}$ denote probability and expectation{, respectively}. {$\Re\{\cdot\}$ and $\Im\{\cdot\}$ denote the real and imaginary parts, respectively, and $\mathrm{diag}(\cdot)$ denotes a diagonal matrix.} $\mathbf{I}$, $\mathbf{1}$, and $\mathbf{e}_i$ denote an identity matrix, an all-one vector, and the $i$-th standard basis vector{, respectively}. Moreover, $[x]^+\triangleq\max\{x,0\}$, $\lceil x\rceil$ is the ceiling operation, {$\inf$ and $\sup$ denote the infimum and supremum, $a\wedge b\triangleq\min\{a,b\}$}, $\mathbb{I}\{\cdot\}$ is the indicator function, $sp(\cdot)$ is the spectral radius  and $\otimes$ denotes min-plus convolution. Finally, $\mathcal{CN}(\boldsymbol{\mu},\boldsymbol{\Sigma})$ denotes a circularly symmetric complex Gaussian distribution {with mean $\boldsymbol{\mu}$ and covariance matrix $\boldsymbol{\Sigma}$}, and $j\triangleq\sqrt{-1}$.

\section{System Model and Problem Formulation}

\subsection{Channel Model}
As shown in Fig.~\ref{fig:system_model}, we consider an MA-{enhanced} MU-MI{M}O downlink system, where a BS equipped with $M$ Tx-MAs serves $K$ users over the same time-frequency resource. Each user is equipped with a single Rx-MA. The sets of Tx-MAs and users are denoted by $\mathcal{M}\triangleq\{1,\ldots,M\}$ and $\mathcal{K}\triangleq\{1,\ldots,K\}$, respectively. All MAs are connected to RF chains via flexible cables. {At the BS, all $M$ Tx-MAs can be repositioned within a {given} two-dimensional (2D) movable region $\mathcal{C}^{\mathrm{s}}\!\subset\!\mathbb{R}^2$, whereas the single Rx-MA of user $k$ can move within its own 2D movable region $\mathcal{C}_k^{\mathrm{r}}\subset\mathbb{R}^2$.} Let $\boldsymbol{\Lambda}_m(t)\!=\!\left[x_m^{\mathrm{s}}(t),y_m^{\mathrm{s}}(t)\right]^T\!\!\!\in\!\mathcal{C}^{\mathrm{s}}$ denote the 2D position of the $m$-th Tx-MA. The Tx-side MA position matrix is defined as $\boldsymbol{\Lambda}(t)\!=\!\begin{bmatrix}\boldsymbol{\Lambda}_1(t),\cdots,\boldsymbol{\Lambda}_M(t)\end{bmatrix}\!\!\in\!\!\mathbb{R}^{2\times M}$. Similarly, the position of the Rx-MA {at} user $k$ is denoted by $\boldsymbol{\Omega}_k(t)=\left[x_k^{\mathrm{r}}(t),y_k^{\mathrm{r}}(t)\right]^T\in\mathcal{C}_k^{\mathrm{r}}$, where $\mathcal{C}_k^{\mathrm{r}}$ is the {user-specific} Rx-side movable region. The Rx-side MA position matrix is $\boldsymbol{\Omega}(t)=\begin{bmatrix}\boldsymbol{\Omega}_1(t),\cdots,\boldsymbol{\Omega}_K(t)\end{bmatrix}\in\mathbb{R}^{2\times K}$. Since multiple Tx-MAs share the {same} BS-side movable region, their positions should be sufficiently separated to avoid mutual coupling, i.e.,
\begin{equation}
\left\|\boldsymbol{\Lambda}_i(t)-\boldsymbol{\Lambda}_j(t)\right\|\geq D,\forall i\neq j, 
\end{equation}
\par \noindent where $D$ is the minimum allowable distance between any two Tx-MAs.

We adopt the field-response channel model {in \cite{zhu2024modeling}} to characterize the position-dependent {BS-user channels}. Given that the BS-side and user-side movable regions are much smaller than the corresponding BS-user propagation distances, the far-field condition holds {throughout these regions}. For {each} user $k$, let $L_{k}^{\mathrm{s}}$ and $L_{k}^{\mathrm{r}}$ denote the numbers of Tx-side and Rx-side channel paths, respectively. The elevation and azimuth angles of the $l^{\mathrm{s}}$-th Tx-side path are denoted by $\theta_{k,l^{\mathrm{s}}}^{\mathrm{s}}$ and $\phi_{k,l^{\mathrm{s}}}^{\mathrm{s}}$, while those of the $l^{\mathrm{r}}$-th Rx-side path are denoted by $\theta_{k,l^{\mathrm{r}}}^{\mathrm{r}}$ and $\phi_{k,l^{\mathrm{r}}}^{\mathrm{r}}$, respectively. The corresponding {angles of departure (AoDs) and angles of arrival (AoAs)} are defined as
\begin{subequations}
\begin{equation}
\begin{aligned}
\vartheta_{k,l^{\mathrm{s}}}^{\mathrm{s}}
&=
\cos\theta_{k,l^{\mathrm{s}}}^{\mathrm{s}}
\cos\phi_{k,l^{\mathrm{s}}}^{\mathrm{s}},
\ \ 
\varphi_{k,l^{\mathrm{s}}}^{\mathrm{s}}
=
\cos\theta_{k,l^{\mathrm{s}}}^{\mathrm{s}}
\sin\phi_{k,l^{\mathrm{s}}}^{\mathrm{s}}, \\
\varpi_{k,l^{\mathrm{s}}}^{\mathrm{s}}
&=
\sin\theta_{k,l^{\mathrm{s}}}^{\mathrm{s}},
\quad\ \ \ \,
1\le k\le K,\quad\ \ \ \, 1\le l^{\mathrm{s}}\le L_k^{\mathrm{s}},
\end{aligned}
\end{equation}
\begin{equation}
\begin{aligned}
\vartheta_{k,l^{\mathrm{r}}}^{\mathrm{r}}
&=
\cos\theta_{k,l^{\mathrm{r}}}^{\mathrm{r}}
\cos\phi_{k,l^{\mathrm{r}}}^{\mathrm{r}},
\ \ 
\varphi_{k,l^{\mathrm{r}}}^{\mathrm{r}}
=
\cos\theta_{k,l^{\mathrm{r}}}^{\mathrm{r}}
\sin\phi_{k,l^{\mathrm{r}}}^{\mathrm{r}}, \\
\varpi_{k,l^{\mathrm{r}}}^{\mathrm{r}}
&=
\sin\theta_{k,l^{\mathrm{r}}}^{\mathrm{r}},
\quad\ \ \ \,
1\le k\le K,\quad\ \ \ \, 1\le l^{\mathrm{r}}\le L_k^{\mathrm{r}}.
\end{aligned}
\end{equation}
\end{subequations}
For the $m$-th Tx-MA, the path-length difference for the $l^{\mathrm{s}}$-th Tx-side path of user $k$ is given by $\rho_{k,l^{\mathrm{s}}}^{\mathrm{s}}\left(\boldsymbol{\Lambda}_{m}(t)\right)=x_{m}^{\mathrm{s}}(t)\vartheta_{k,l^{\mathrm{s}}}^{\mathrm{s}}+y_{m}^{\mathrm{s}}(t)\varphi_{k,l^{\mathrm{s}}}^{\mathrm{s}}.$ Similarly, the path-length difference induced by the Rx-MA position of user $k$ for the $l^{\mathrm{r}}$-th Rx-side path is $\rho_{k,l^{\mathrm{r}}}^{\mathrm{r}}\left(\boldsymbol{\Omega}_k(t)\right)=x_k^{\mathrm{r}}(t)\vartheta_{k,l^{\mathrm{r}}}^{\mathrm{r}}+y_k^{\mathrm{r}}(t)\varphi_{k,l^{\mathrm{r}}}^{\mathrm{r}}.$ Accordingly, the transmit field-response vector is
\begin{equation}
\mathbf{g}_k\!\left(\!\mathbf{\Lambda}_m(t)\right)\!
=\!\!
\left[
e^{j\!\frac{2\pi}{\lambda}\!\rho_{k,1}^{\mathrm{s}}\!\left(\!\mathbf{\Lambda}_m(t)\right)}\!,
\!\cdots\!,\!
e^{j\!\frac{2\pi}{\lambda}\!\rho_{k,L_k^{\mathrm{s}}}^{\mathrm{s}}\!\left(\!\mathbf{\Lambda}_m(t)\right)}
\right]^{\!T}\!\!\!\!\in\!\mathbb{C}^{L_k^{\mathrm{s}}\!\times\!1} {\!,}
\end{equation} {where $\lambda$ denotes the carrier wavelength}. Stacking the transmit field-response vectors of all Tx-MAs yields the transmit field-response matrix of user $k$ as
\begin{equation}
\mathbf{G}_k\left(\boldsymbol{\Lambda}(t)\right)=\left[\mathbf{g}_k\left(\boldsymbol{\Lambda}_1(t)\right),\ldots,\mathbf{g}_k\left(\boldsymbol{\Lambda}_M(t)\right)\right]\in\mathbb{C}^{L_k^{\mathrm{s}}\times M}.
\end{equation}
The receive field-response vector of user $k$ is defined as
\begin{equation}
\mathbf{f}_k\!\left(\boldsymbol{\Omega}_k(t)\right)\!\!=\!\!\left[e^{j\!\frac{2\pi}{\lambda}\!\rho_{k,1}^{\mathrm{r}}\!(\boldsymbol{\Omega}_k(t))}\!,\!\cdots\!,e^{j\!\frac{2\pi}{\lambda}\!\rho_{k,L_k^{\mathrm{r}}}^{\mathrm{r}}\!(\boldsymbol{\Omega}_k(t))}\right]^{\!T}\!\!\!\in\!\mathbb{C}^{L_k^{\mathrm{r}}\!\times\!1}\!.
\end{equation} Let $\boldsymbol{\Sigma}_k(t)\in\mathbb{C}^{L_k^{\mathrm{r}}\times L_k^{\mathrm{s}}}$ denote the path-response matrix of user $k$, where the $(l^{\mathrm{r}},l^{\mathrm{s}})$-th element characterizes the complex response coefficient between the $l^{\mathrm{s}}$-th Tx-side path and the $l^{\mathrm{r}}$-th Rx-side path. Based on the above definitions, the equivalent downlink channel vector of user $k$ is given by
\begin{equation}
\mathbf{h}_k^{{H}}\!\!\left(\boldsymbol{\Lambda}(t),\boldsymbol{\Omega}_k(t)\right)\!=\!\mathbf{f}_k^{{H}}\!\!\left(\boldsymbol{\Omega}_k(t)\right)\!\boldsymbol{\Sigma}_k(t)\mathbf{G}_k\!\left(\boldsymbol{\Lambda}(t)\right)\in\mathbb{C}^{1\times M}.
\end{equation}
By stacking all {BS-user} channels, the overall downlink channel matrix is obtained as
\vspace{-0.05cm}
\begin{equation}
\mathbf{H}\left(\boldsymbol{\Lambda}(t),\boldsymbol{\Omega}(t)\right)=\begin{bmatrix}\mathbf{h}_1^{{H}}\left(\boldsymbol{\Lambda}(t),\boldsymbol{\Omega}_1(t)\right)\\\vdots\\\mathbf{h}_K^{{H}}\left(\boldsymbol{\Lambda}(t),\boldsymbol{\Omega}_K(t)\right)\end{bmatrix}\in\mathbb{C}^{K\times M}.
\end{equation}

It follows that the BS-user channels in the MA-{enhanced} MU-MIMO system are jointly determined by the Tx-side position matrix $\boldsymbol{\Lambda}(t)$ and the Rx-side position matrix $\boldsymbol{\Omega}(t)$.

\begin{figure*}[!t]
    \centering
    \vspace{-0.7cm}
    \includegraphics[width=17.5cm]{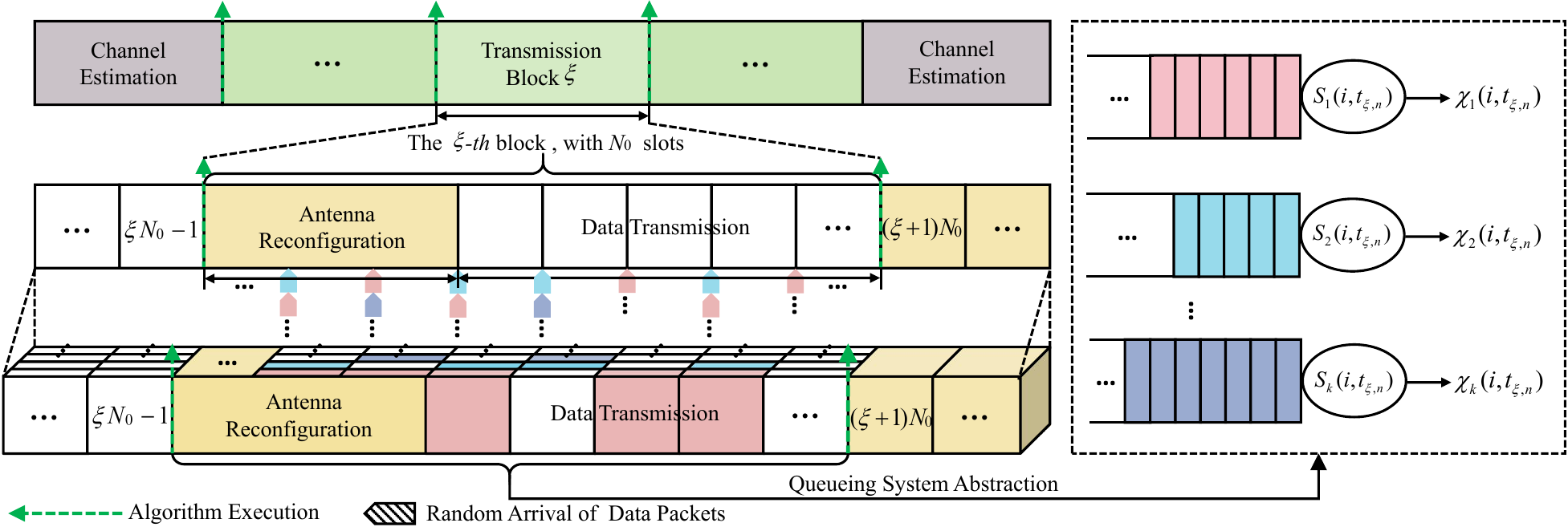}
    \vspace{-0.3cm}
    \caption{{Two-timescale protocol and queueing model for the considered system.}}
    \vspace{-0.3cm}
    \label{fig:transmission protocol}
\end{figure*} 

\subsection{Two-Timescale {Protocol} and Block Service Model}

We consider a {block-fading channel model, where all associated channel parameters are} assumed to remain quasi-static within each block, while {they} may vary between {different} blocks. Since MA movement relies on mechanical, electromagnetic, or fluidic actuation mechanisms, antenna positions cannot be updated instantaneously at the packet-arrival timescale. We therefore adopt a two-timescale transmission protocol, as illustrated in 
Fig.~\ref{fig:transmission protocol}. The transmission timeline is divided into consecutive blocks, each comprising $N_{0}$ time slots of duration $\tau${; hence,} the block duration is $T=N_{0}\tau$. MA positions and transmit beamforming vectors are updated at a large (block) timescale, whereas packet arrivals and queues evolve at a small (slot) timescale. Let $\boldsymbol{\Xi}$ denote the set of block indices, with $\xi\in\boldsymbol{\Xi}$ being the block index. The global slot set of block $\xi$ is {defined as} $\mathcal{T}_{\xi}\triangleq \{\xi N_{0},\xi N_{0}+1,\ldots,(\xi+1)N_{0}-1\}$, and its $n$-th local slot is denoted by $t_{\xi,n}\triangleq\xi N_{0}+n$. {Within} each slot, packets arrive randomly at the BS buffer, while at the beginning of block $\xi$, the controller determines the target MA positions and beamforming strategy. Let {$\boldsymbol{\Lambda}_m(t_{\xi,n})$ and $\boldsymbol{\Omega}_k(t_{\xi,n})$} denote the positions of the {$m$-th Tx-MA and the Rx-MA of user $k$} at the {$n$-th time slot in block $\xi$}, respectively. Then, within the {$\xi$-th time block}, the $m$-th Tx-MA is moved from $\boldsymbol{\Lambda}_{m}(t_{\xi,0})$ to $\boldsymbol{\Lambda}_{m}(t_{\xi,{n_{\mathrm{s},m}}})$, and the Rx-MA of user $k$ is moved from $\boldsymbol{\Omega}_{k}(t_{\xi,0})$ to $\boldsymbol{\Omega}_{k}(t_{\xi,{n_{\mathrm{r},k}}})$, where ${n_{\mathrm{s},m}}$ and ${n_{\mathrm{r},k}}$ denote the {indices of the time slots} when they move to their respective {target positions}, respectively. {Let} the Tx-MA and Rx-MA {movement speeds be denoted as} $v^{\mathrm{s}}$ and $v^{\mathrm{r}}$, respectively{. Since the MAs at the BS and user sides can be moved} in parallel, the {effective movement delays} at the BS and user sides are determined by {the largest delay among all Tx-MAs and {Rx-MAs}}, respectively, i.e.,
\vspace{-0.1cm}
\begin{subequations}\label{eq:movement-delay}
\begin{align}
W^{\mathrm{s}}(\xi)
&=\max_{m\in\mathcal{M}}
\frac{
\left\|\boldsymbol{\Lambda}_m(t_{\xi,{n_{\mathrm{s},m}}})
-\boldsymbol{\Lambda}_m(t_{\xi,0})\right\|_2
}{v^{\mathrm{s}}},
\label{eq:tx-movement-delay}\\
W^{\mathrm{r}}(\xi)
&\,=\,\max_{k\in\mathcal{K}}\,
\frac{
\left\|\boldsymbol{\Omega}_k(t_{\xi,{n_{\mathrm{r},k}}})\,
-\,\,\boldsymbol{\Omega}_k(t_{\xi,0})\right\|_2
}{v^{\mathrm{r}}}.
\end{align}
\vspace{-0.55cm}
\end{subequations}
\par \noindent {It follows that} the overall movement delay {for the $\xi$-th time block} is {given by} $ W_{\mathrm{mov}}(\xi)\!\!=\!\!\max\!\begin{Bmatrix}\!W^{\mathrm{s}}(\xi),\!W^{\mathrm{r}}(\xi)\!\end{Bmatrix}\!\!.$
Accordingly, the number of slots occupied by MA {movement} is 
\begin{equation}
{N_{\mathrm{mov}}}(\xi)
=
\min
\left\{
N_{0},
\left\lceil
{W_{\mathrm{mov}}(\xi)}/{\tau}
\right\rceil
\right\},
\label{slot rounding}
\end{equation}
and the number of slots available for data transmission {(DT) in the $\xi$-th time block} is $N_{\mathrm{DT}}(\xi)=N_{0}-{N_{\mathrm{mov}}}(\xi).$
The local slot sets corresponding to the {MA movement} phase and the data-transmission phase are respectively defined as ${\mathcal{T}_{\mathrm{mov}}}(\xi)=
\{n:0\le n<{N_{\mathrm{mov}}}(\xi)\},$ and $\mathcal{T}_{\mathrm{DT}}(\xi) = \{n:{N_{\mathrm{mov}}}(\xi)\le n<N_{0}\}.$ During ${\mathcal{T}_{\mathrm{mov}}}(\xi)$,  the equivalent channel may fluctuate and reliable short-packet transmission cannot be guaranteed. Therefore, no {DT} is performed during {this} phase, whereas downlink {DT} is {performed} during $\mathcal{T}_{\mathrm{DT}}(\xi)$. 
After MA {movement} is completed in block $\xi$, the BS-side Tx-MAs and user-side Rx-MAs remain fixed during the DT phase. These positions carry over to the next block as initial positions, i.e., 
\vspace{-0.45cm}
\begin{subequations}
\label{eq:MA_position_continuity_model}
\begin{align}
\boldsymbol{\Lambda}_m(t_{\xi+1,0})
&=
\boldsymbol{\Lambda}_m(t_{\xi,n_{\mathrm{s},m}}),
\quad \forall m,\xi,
\label{eq:Tx_position_continuity_model}\\
\boldsymbol{\Omega}_k(t_{\xi+1,0})
&=
\boldsymbol{\Omega}_k(t_{\xi,n_{\mathrm{r},k}}),
\quad \forall k,\xi.
\label{eq:Rx_position_continuity_model}
\end{align}
\end{subequations}
Based on the field-response channel model {presented} in Section~II-A, the downlink channel vector of user $k$ is denoted by $\mathbf{h}_k(\xi)=\mathbf{h}_k\left(\{\boldsymbol{\Lambda}_m(t_{\xi,{n_{\mathrm{s},m}}})\}_{m\in\mathcal{M}},\boldsymbol{\Omega}_k(t_{\xi,{n_{\mathrm{r},k}}})\right).$ Let $\boldsymbol{\omega}_{k}(\xi)\in\mathbb{C}^{M\times 1}$ denote the beamforming vector for user $k$ in block $\xi$, and define the {BS precoding} matrix as $\boldsymbol{\omega}(\xi)
=
[\boldsymbol{\omega}_{1}(\xi),
\ldots,
\boldsymbol{\omega}_{K}(\xi)]
\in\mathbb{C}^{M\times K}.$ Then, the baseband {received signal at} user $k$ at any DT slot $t_{\xi,n}$ is given by
\vspace{-0.05cm}
\begin{equation}
\begin{aligned}
& y_{k}(t_{\xi,n}) 
\,=\, \mathbf{h}_{k}^{H}(\xi)\,\boldsymbol{\omega}_{k}(\xi)\,x_{k}(t_{\xi,n})\, \,+ \\
&\!\!\!\!\! \vspace{-0.5cm}\!\sum_{\substack{j=1,j\ne k}}^{K}\!\!\!
\mathbf{h}_{k}^{H}(\xi)\boldsymbol{\omega}_{j}(\xi)x_{j}(t_{\xi,n})
+ n_{k}(t_{\xi,n}),
\end{aligned}
\vspace{-0.05cm}
\end{equation} 
where $x_{k}(t_{\xi,n})$ is the normalized information symbol intended for user $k$ with $\mathbb{E}{|x_{k}(t_{\xi,n})|^{2}}=1$, and $n_{k}(t_{\xi,n})\sim\mathcal{CN}(0,\sigma_{k}^{2})$ denotes the receiver noise{, with $\sigma_k^2$ denoting the average noise power}. The {signal-to-interference-plus-noise ratio (SINR)} {achievable at} user $k$ in block $\xi$ is expressed as
\begin{equation}
\gamma_{k}(\xi)
=
\frac{
\left|
\mathbf{h}_{k}^{H}(\xi)\boldsymbol{\omega}_{k}(\xi)
\right|^{2}
}{
\sum_{j=1,j\ne k}^{K}
\left|
\mathbf{h}_{k}^{H}(\xi)\boldsymbol{\omega}_{j}(\xi)
\right|^{2}
+
\sigma_{k}^{2}
}.
\end{equation} 
{We consider finite-blocklength transmission {in this paper}, where each codeword spans $m_{\mathrm{s}}$ channel uses.} Let $B$ denote the bandwidth. The finite-blocklength achievable rate of user $k$ during the DT phase of block $\xi$ is {approximated as \cite{he2024shortpacket}}
\vspace{-0.05cm}
\begin{equation}
C_{k}(\xi)
\!\approx\!
B\!\!
\left[
\!\log_{2}(1\!\!+\!\!\gamma_{k}(\xi))
\!\!-\!\!
(\log_{2}\!e)\!
\sqrt{\!
\frac{(2\!\!+\!\gamma_{k}(\xi))\gamma_{k}(\xi)}
{\!\!(1\!\!+\!\gamma_{k}(\xi))^{2}{m_{\mathrm{s}}}}
}\!
f_{Q}^{-1}\!(\hbar)
\!\right]^{\!\!+}\!,
\end{equation}
\vspace{-0.4cm} \par \noindent where $[x]^+\triangleq\max\{x,0\}$, {$\hbar$ denotes the prescribed error rate} and $f_{Q}^{-1}(\cdot)$ denotes the inverse Gaussian $Q$-function. {The} per-slot {achievable rate} of user $k$ is {thus given by} $\tau C_k(\xi)$. {To keep consistent with the queuing model presented later, we normalize this rate in number of bits as the number of packets by defining $\hat{C}_{k}(\xi)=\tau C_k(\xi)/L_{\mathrm{pkt}}$, where $L_{\mathrm{pkt}}$ denotes the number of bits in one packet.}

\vspace{-0.3cm}
\subsection{Queue Evolution and Delay Model}
\vspace{-0.08cm}
{We {assume} that the BS maintains a {first-in-first-out (FIFO)} queue for each user, with arrival, service, and backlog measured in packets.} To characterize the bursty nature of delay-sensitive traffic, we model the packet arrival process as a Markov-modulated On-Off (MMOO) process, as illustrated in Fig.~\ref{fig:MMOO_model}. {Let $Z_{k}(t_{\xi,n})$ denote the traffic state of user $k$ at slot $t_{\xi,n}$, {which takes the value of} 0 {or} 1, {corresponding} to the OFF (no arrival) and ON (bursty arrival) states, respectively.} The state transition matrix is given by
\vspace{-0.025cm}
\begin{equation}
\mathbf{P}_k=
\begin{bmatrix}
1-p_k & p_k\\
q_k & 1-q_k
\end{bmatrix},
\end{equation}

\par \noindent where $p_k$ and $q_k$ denote the OFF-to-ON and ON-to-OFF transition probabilities, respectively. Accordingly, the steady-state probabilities of the OFF and ON states are {given by} $\pi_{k,\mathrm{off}}=q_k/ \!\left(p_k\!+\!q_k\right),\pi_{k,\mathrm{on}} = p_k/ \!\left(p_k\!+\!q_k\right)\!,$ respectively.

\begin{figure}[t] 
\centering
\vspace{-0.5cm}
\includegraphics[width=7.5cm]
{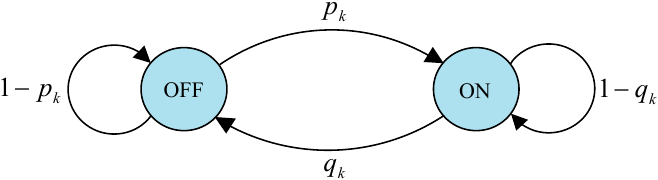} 
\vspace{-0.3cm}
\caption{Illustration of the {MMOO} process.}
\vspace{0.05cm}
\label{fig:MMOO_model} 
\end{figure}

Let $R_k$ denote the number of packets arriving in the ON state and {$a_k(t_{\xi,n})$} denote the packet arrival amount in slot $t_{\xi,n}$, {with} $a_k(t_{\xi,n})\,=\,Z_k(t_{\xi,n})\,R_k.$ {Hence,} the cumulative arrival
\par \noindent {from slot $i$ to slot $j$} within block $\xi$ is
\vspace{-0.08cm}
\begin{equation}
    A_{k}(t_{\xi,i},t_{\xi,j})\!=\!\sum_{n=i}^{j-1}a_{k}(t_{\xi,n}). 
    \vspace{-0.06cm}
\end{equation} Let {$s_k(t_{\xi,n})$} {denote} the number of packets that can be served for user $k$ in slot $t_{\xi,n}$. Since no effective service is available during MA movement, $s_k(t_{\xi,n})=\mathbb{I}\{n\!\in\!\mathcal{T}_{\mathrm{DT}}(\xi)\}\hat{C}_k(\xi),$ where {$\mathbb{I}\{\cdot\}$ denotes the indicator function}. Accordingly, the cumulative service provided to user $k$ over $[t_{\xi,i},t_{\xi,j})$ is
\vspace{-0.15cm}
\begin{equation}
S_k(t_{\xi,i},t_{\xi,j})
\!=\!\!\sum_{n=i}^{j-1}s_k(t_{\xi,n}).
\vspace{-0.05cm}
\end{equation}
Let $Q_k(t_{\xi,n})$ denote the queue backlog of the $n$-th {time} slot in block $\xi$. Its evolution is given by 
\begin{equation}
\vspace{-0.08cm}
Q_k(t_{\xi,n+1}) \!=\! \left[Q_k(t_{\xi,n})\!-\!s_k(t_{\xi,n})\right]^+
\!+a_k(t_{\xi,n}).
\vspace{-0.05cm}
\end{equation}
Since the queue evolves continuously across transmission blocks, {we have} $Q_k(t_{\xi+1,0})\!=\!Q_k(t_{\xi,N_0}).$ As no effective data service is provided during MA {movement}, the backlog at the  {beginning} of the DT interval is expressed as
\begin{equation}
\vspace{-0.08cm}
{B_k(\xi)}
=
Q_k(t_{\xi,0})
+
A_k\!\left(t_{\xi,0},
t_{\xi,{N_{\mathrm{mov}}}(\xi)}\right),
\vspace{-0.05cm}
\label{eq:start_backlog}
\end{equation}
which comprises the inherited backlog i.e., $Q_k(t_{\xi,0})$ and the arrivals accumulated during {antenna movement} i.e., $A_k\!\left(t_{\xi,0},
t_{\xi,{N_{\mathrm{mov}}}(\xi)}\right)$.

For notational compactness, define the DT-phase {time-slot index as}
$\tilde{t}_{\xi,r}\triangleq
t_{\xi,N_{\mathrm{mov}}(\xi)+r}$,
where $r=0,\ldots,N_{\mathrm{DT}}(\xi)$.
{Based on \eqref{eq:start_backlog},} the total traffic demand to be
served {in} DT-phase slot $r$ is given by
\begin{equation}
\vspace{-0.05cm}
A_{k,\xi}^{\mathrm{DT}}(0,r)
\triangleq
B_k(\xi)
+
A_k(\tilde{t}_{\xi,0},\tilde{t}_{\xi,r}),
\quad r\geq0,
\label{eq:DT_augmented_arrival}
\end{equation}
with $A_{k,\xi}^{\mathrm{DT}}(0,r)=0$ for $r<0$. For
$0\leq u\leq r\leq N_{\mathrm{DT}}(\xi)$, the cumulative service provided during the DT phase is
\begin{equation}
\vspace{-0.05cm}
S_{k,\xi}^{\mathrm{DT}}(u,r)
\triangleq
S_k(\tilde{t}_{\xi,u},\tilde{t}_{\xi,r})
=
(r-u)\hat{C}_k(\xi).
\vspace{-0.05cm}
\label{eq:DT_service_curve}
\end{equation}
Under the FIFO discipline, the cumulative departure
satisfies the standard min-plus input--output relation
\cite{poloczek2014envelope}, i.e.,
\begin{equation}
\vspace{-0.05cm}
\chi_{k,\xi}^{\mathrm{DT}}\!(0,\!r)
\!\!=\!\!
\left(
A_{k,\xi}^{\mathrm{DT}}
\otimes
S_{k,\xi}^{\mathrm{DT}}
\right)\!(0,\!r)\!\!=\!\!\!\!\!
\inf_{0\leq u\leq r}\!\!\!
\left\{\!
A_{k,\xi}^{\mathrm{DT}}\!(0,u)
\!\!+\!\!
S_{k,\xi}^{\mathrm{DT}}\!(u,r)\!
\right\}\!.
\vspace{-0.05cm}
\label{eq:DT_departure}
\end{equation} The DT-phase queueing delay at $\tilde{t}_{\xi,r}$ is defined as the minimum backward interval for which the cumulative traffic demand does not exceed the cumulative departures \cite{poloczek2015service}, i.e.,
\begin{equation}
\vspace{-0.05cm}
D_k^{\mathrm{queue}}(\tilde{t}_{\xi,r})
\!=\!
\inf\!\left\{
 j \geq 0\!:\!
 A_{k,\xi}^{\mathrm{DT}}(0,r\!-\!j)
 \!\leq\!
 \chi_{k,\xi}^{\mathrm{DT}}(0,r)
\right\}\!.
\vspace{-0.05cm}
\label{eq:queueing-delay definition}
\end{equation} Given {a} prescribed maximum delay threshold $D_k^{\max}$, the delay
budget remaining for DT-phase queueing is
\vspace{-0.08cm}
\begin{equation}
D_k^{\mathrm{rem}}(\xi)
=
D_k^{\max}
-
N_{\mathrm{mov}}(\xi).
\label{eq:remaining_delay_budget}
\vspace{-0.08cm}
\end{equation}
Let $\varepsilon_k$ denote a prescribed maximum delay-outage probability. Then, it must hold that
\vspace{-0.1cm}
\begin{equation}
\Pr\!\left\{
D_k(\tilde{t}_{\xi,r})
>
D_k^{\max}
\right\}
\leq
\varepsilon_k,
\quad k\in\mathcal{K}.
\label{eq:delay_outage_constraint}
\vspace{-0.05cm}
\end{equation}

\vspace{-0.1cm}
\subsection{{Problem} {Formulation}}
For block $\xi$, define the target position sets of the BS-side Tx-MAs and user-side Rx-MAs as
\begin{equation}
\vspace{-0.1cm}
\boldsymbol{\Lambda}(\xi)
\!=\!
\left\{\!
\boldsymbol{\Lambda}_m\!\!
\left(
t_{\xi,{n_{\mathrm{s},m}}}\!
\right)\!
\right\}_{\!m\in\mathcal{M}}\!,
\boldsymbol{\Omega}(\xi)
\!=\!
\left\{\!
\boldsymbol{\Omega}_k\!\!
\left(
t_{\xi,{n_{\mathrm{r},k}}}
\right)\!
\right\}_{\!k\in\mathcal{K}}\!.
\label{eq:target_position_sets}
\end{equation}
In this paper, we aim to minimize the long-term average transmit power while satisfying the delay-outage probability constraints and practical MA movement constraints. The associated optimization problem is formulated as\vspace{-0.06cm}
\begin{subequations}
\begin{align}
\optprob{prob:P1}:\!\!
\min_{\{\boldsymbol{\Lambda}(\xi),\boldsymbol{\Omega}(\xi),\boldsymbol{\omega}(\xi)\}}
&
\lim_{|\boldsymbol{\Xi}|\rightarrow\infty}\!\!\frac{1}{|\boldsymbol{\Xi}|}\!\!
\sum_{\xi\in \boldsymbol{\Xi}}
\mathrm{Tr}\!\left(\boldsymbol{\omega}(\xi)\boldsymbol{\omega}^{H}(\xi)\right)\label{eq:P1_obj}  \\
\mathrm{s.t.}\quad
&\!\!\!\!\!\Pr\!\left\{D_k(\tilde{t}_{\xi,r})\!>\!D_k^{\max}\right\}
\!\leq\! \varepsilon_k,\forall k, \label{eq:P1_delay} \\
& \!\!\!\!\boldsymbol{\Lambda}_m\!\!\left(t_{\xi,{n_{\mathrm{s},m}}}\!\right)\in\mathcal{C}^{\mathrm{s}},
\forall m,\xi, \label{eq:P1_Tx_region} \\
& \!\!\!\!\boldsymbol{\Omega}_k\left(t_{\xi,{n_{\mathrm{r},k}}}\right)\in\mathcal{C}^{\mathrm{r}}_k,
\forall k,\xi, \label{eq:P1_Rx_region} \\
& \!\!\!\!\|\boldsymbol{\Lambda}_i(\xi)\!\!-\!\!\boldsymbol{\Lambda}_j(\xi)\|_2\geq D,\forall i\neq j,\xi, \label{eq:P1_spacing} \\
&
\!\!\!\! \boldsymbol{\Lambda}_m(t_{\xi+1,0})=\boldsymbol{\Lambda}_m\!(t_{\xi,{n_{\mathrm{s},m}}}),
\forall m,\xi, \label{eq:P1_Tx_continuity}\\
&
\!\!\!\! 
\boldsymbol{\Omega}_k(t_{\xi+1,0})
=
\boldsymbol{\Omega}_k(t_{\xi,{n_{\mathrm{r},k}}}),
\forall k,\xi, \label{eq:P1_Rx_continuity}\\
&
\begin{aligned}[t]
& 
\!\!\!\! 
Q_k(t_{\xi+1,0})
\!\!=\!\!
\Big[Q_k(t_{\xi,0})
\!+\!A_k(t_{\xi,0},t_{\xi+1,0})
\\[-0.8ex]
&
\!\!\!\! 
-S_k(t_{\xi,0},t_{\xi+1,0})\Big]^+,
\forall k,\xi.
\end{aligned}
\raisetag{3.5ex}
\label{eq:P1_queue}
\end{align}
\end{subequations} \vspace{-0.25cm}\par\noindent In \probref{prob:P1}, constraint \eqref{eq:P1_delay} imposes the statistical delay-outage requirement in \eqref{eq:delay_outage_constraint}, constraints \eqref{eq:P1_Tx_region}–\eqref{eq:P1_spacing} specify the maximum MA movement regions and the minimum spacing between Tx-MAs, constraints \eqref{eq:P1_Tx_continuity}–\eqref{eq:P1_Rx_continuity} enforce the position continuity in \eqref{eq:MA_position_continuity_model}, and constraint \eqref{eq:P1_queue} characterizes the queue evolution across consecutive transmission blocks.

However, problem ~\probref{prob:P1} is challenging to solve optimally due to its stochastic cross-layer and temporal coupling. The probabilistic delay constraint in \eqref{eq:P1_delay} depends not only on the current physical-layer service rate, but also on MAC-layer traffic arrivals, inherited queue backlog, and DT-phase queue evolution. 
Moreover, \probref{prob:P1} is temporally coupled, as the objective is the long-term
average power, whereas MA positions and queue backlog in different blocks are
coupled, as seen from \eqref{eq:P1_Tx_continuity}–\eqref{eq:P1_queue}. Next, we first address the probabilistic delay-outage constraints in \eqref{eq:P1_delay} by developing a deterministic sufficient condition.

\vspace{0.2cm}
\section{Delay-Outage Analysis}
\vspace{-0.1cm}
To tackle the probabilistic delay-outage constraints,
we first introduce several concepts and basic properties
of martingales that will be used in the subsequent derivations. Let $\{\mathcal{F}_n\}_{n\geq 0}$ be a filtration, i.e., an increasing sequence of $\sigma$-algebras satisfying $\mathcal{F}_n\!\subseteq\!\mathcal{F}_{n+1}$ for all $n\!\geq \!0$.
A stochastic process $\{X(n)\}_{n\geq 0}$ is said to be adapted to $\{\mathcal{F}_n\}_{n\geq 0}$ if
$X(n)$ is $\mathcal{F}_n$-measurable for every $n\geq 0$.

\begin{definition}[Martingale] Let $\{X(n)\}_{n\geq 0}$ be adapted to $\{\mathcal{F}_n\}_{n\geq 0}$. If it satisfies
\begin{equation}
\mathbb{E}\!\left[|X(n)|\right]<\infty,
\,\,\,
\mathbb{E}\!\left[X(n+1)\mid\mathcal{F}_n\right]=X(n),
\,\,\, \forall n\geq0.
\label{eq:basic_martingale_definition}
\end{equation}
then $\{X(n)\}_{n\geq 0}$ is called a martingale.
\label{def:martingale_preliminary}
\end{definition}

We will also use the notions of a stopping time and the optional
stopping theorem. A random variable
$\tau\in\mathbb{Z}_{\geq 0}\cup\{\infty\}$ is a stopping time
with respect to $\{\mathcal{F}_n\}_{n\geq 0}$ if
$\{\tau\leq n\}\in\mathcal{F}_n$ for every $n\geq 0$.
For a martingale $\{X(n)\}_{n\geq 0}$ and any bounded stopping
time $\tau$, the optional stopping theorem gives
\begin{equation}
    \mathbb{E}[X(\tau)]=\mathbb{E}[X(0)].
\end{equation}

Based on these definitions and properties, we next construct a martingale
for the downlink queue backlog and use it to derive the
delay-outage probability.

\begin{figure*}[!b]
\hrule
\vspace{-15pt}
\noindent
\vspace{0pt}

\begin{flalign*}
&\Pr\!\left\{
D_k^{\mathrm{queue}}(\tilde{t}_{\xi,r})
>D_k^{\mathrm{rem}}(\xi)
\right\}
\overset{\mathrm{(a)}}{=}
\Pr\!\left\{
B_k(\xi)
+A_k\!\left(
\tilde{t}_{\xi,0},
\tilde{t}_{\xi,r-D_k^{\mathrm{rem}}(\xi)}
\right)
\geq
\chi_k\!\left(
\tilde{t}_{\xi,0},
\tilde{t}_{\xi,r}
\right)
\right\}
&&\\
&\!\overset{\mathrm{(b)}}{\leq}\!
\Pr\!\left\{
\sup_{0\leq r'\leq r-D_k^{\mathrm{rem}}(\xi)}
\left[
B_k(\xi)
+A_k\!\left(
\tilde{t}_{\xi,r'},
\tilde{t}_{\xi,r-D_k^{\mathrm{rem}}(\xi)}
\right)
-S_k\!\left(
\tilde{t}_{\xi,r'},
\tilde{t}_{\xi,r}
\right)
\right]
\geq 0
\right\}
&&\\
&\!\overset{\mathrm{(c)}}{=}\!
\Pr\!\left\{
\sup_{0\leq r'\leq r-D_k^{\mathrm{rem}}(\xi)}
\left[
B_k(\xi)
+A_k\!\left(
\tilde{t}_{\xi,r'},
\tilde{t}_{\xi,r-D_k^{\mathrm{rem}}(\xi)}
\right)
-\left(r-r'\right)\hat{C}_k(\xi)
\right]
\geq 0
\right\}, \,\,\, \mathrm{let}\,\,\, u=r-r',
&&\\
&\!\overset{\mathrm{(d)}}{=}\!
\Pr\!\left\{\!
\sup_{D_k^{\mathrm{rem}}\!(\xi)\leq u\leq r}\!\!
\left[
B_k(\xi)
\!+\!\!A_k\!\!\left(\!
\tilde{t}_{\xi,r-u},
\tilde{t}_{\xi,r-\!D_k^{\mathrm{rem}}\!(\xi)}
\!\right)
\!\!-\!u\hat{C}_{\!k}(\xi)
\right]
\!\!\geq \!0\!
\right\}
\! \overset{\mathrm{(e)}}{\leq} \!
\Pr\!\left\{\!
\sup_{u\geq D_k^{\mathrm{rem}}\!(\xi)}\!\!
\left[
B_k(\xi)
\!+\!\!A_k\!\!\left(\!
\tilde{t}_{\xi,D_k^{\mathrm{rem}}(\xi)},
\tilde{t}_{\xi,u}\!
\right)
\!\!-\!u\hat{C}_{\!k}(\xi)
\right]
\!\!\geq \!0\!
\right\}
&&\\
&\!\overset{\mathrm{(f)}}{=}\!
\Pr\!\left\{\!
\sup_{u\geq D_k^{\mathrm{rem}}(\xi)}
\left[
B_k(\xi)
+A_{k}\!\left(
\tilde{t}_{\xi,D_k^{\mathrm{rem}}(\xi)},
\tilde{t}_{\xi,u}\!
\right)
\!-\!\left(u\!-\!D_k^{\mathrm{rem}}(\xi)\right)\hat{C}_{k}(\xi)
\right]
\! \geq\!
D_k^{\mathrm{rem}}(\xi)\hat{C}_k(\xi)
\right\} , \,\,\, \mathrm{let} \,\,\,  v=u-D_k^{\mathrm{rem}}(\xi),
&&\\
&\!\overset{\mathrm{(g)}}{=}\!
\Pr\!\left\{
\sup_{v\geq 0}
\left[
B_k(\xi)
+
A_k\left(
\tilde{t}_{\xi,D_k^{\mathrm{rem}}(\xi)},
\tilde{t}_{\xi,D_k^{\mathrm{rem}}(\xi)+v}
\right)
-
v\hat{C}_{k}(\xi)
\right]
\geq
D_k^{\mathrm{rem}}(\xi)\hat{C}_{k}(\xi)
\right\}
&&\\
&\!\overset{\mathrm{(h)}}{=}
\Pr\!\left\{
\sup_{r\geq 0}
\left[
B_k(\xi)
+
A_k\left(
\tilde{t}_{\xi,0},
\tilde{t}_{\xi,r}
\right)
-
r\hat{C}_{\!k}(\xi)
\right]
\geq
D_k^{\mathrm{rem}}(\xi)\hat{C}_{k}(\xi)
\right\}.
&&
\tag{\getrefnumber{eqnum:delay_event_to_queue_pressure}}
\label{eq:delay_event_to_queue_pressure}
\end{flalign*}
\end{figure*}

\vspace{-0.15cm}
\subsection{{Backlog Martingale} Analysis}
\vspace{-0.05cm}
Using the backlog $B_k(\xi)$ defined in \eqref{eq:start_backlog}, we characterize
the stochastic evolution of the downlink queue through the associated net-input process
$B_k(\xi)+A_k(\tilde{t}_{\xi,0},\tilde{t}_{\xi,r})
-r\hat{C}_k(\xi)$.
Based on this process, the following proposition constructs a martingale for the subsequent delay-outage analysis.

\begin{proposition}[Backlog Martingale]
For user $k$ in block $\xi$, there exist a positive constant
$\zeta_k^\star$ and a positive function $\Phi_k(\cdot)$ such that the process
\vspace{-0.05cm}
\begin{equation}
M_{k}^{L}(\tilde{t}_{\xi,r})
\!\!=\!\!
\frac{\Phi_{k}(Z_k(\tilde{t}_{\xi,r}))}
{\Phi_{k}(Z_k(\tilde{t}_{\xi,0}))}
e^{\zeta_{k}^{\star}
\left[
{B_k(\xi)}
+
A_k(\tilde{t}_{\xi,0},\tilde{t}_{\xi,r})
-
r\hat{C}_{k}(\xi)
\right]},
\vspace{-0.05cm}
\label{eq:backlog_martingale}
\end{equation}
is a martingale adapted to the filtration
$\{\mathcal{F}_{\tilde{t}_{\xi,r}}\}_{r\geq0}$.
\label{prop:backlog_martingale}
\end{proposition}

\begin{proof}
According to Definition~\ref{def:martingale_preliminary}, it suffices to verify the one-step martingale condition
\begin{equation}
\vspace{-0.05cm}
\mathbb{E}
\left[
M_k^L(\tilde{t}_{\xi,r+1})
\mid
\mathcal{F}_{\tilde{t}_{\xi,r}}
\right]
=
M_k^L(\tilde{t}_{\xi,r}).
\label{eq:one_step_martingale_condition}
\vspace{-0.05cm}
\end{equation}

Conditioned on the current state
$Z_{k}(\tilde{t}_{\xi,r})=i$, the transition probability of the MMOO process is $\Pr\!\left\{
Z_k(\tilde{t}_{\xi,r+1})\!=\!j\mid Z_k(\tilde{t}_{\xi,r})\!=\!i\right\}\!=\![\mathbf{P}_{\!k}]_{i,j}$ for $i,j \!\in\!\{0,1\}$.
Define $f_{k}(j)\triangleq jR_k$, $j\in\{0,1\}$, so that the corresponding one-slot net-input increment is
$f_{k}(j)-\hat{C}_k(\xi)$. Therefore,
\vspace{-0.05cm}
\begin{equation}
\begin{aligned}
&\mathbb{E}
\left[
M_k^L(\tilde{t}_{\xi,r+1})
\mid
\mathcal{F}_{\tilde{t}_{\xi,r}}
\right]
\\
\!\!&=\!\!
\frac{
e^{\zeta_k
\!\left[
{B_{k}\!(\xi)}
+\!
A_k\!(\tilde{t}_{\xi,0},\tilde{t}_{\xi,r})
-
r\hat{C}_{\!k}\!(\xi)
\right]}
}{
\Phi_k(Z_k(\tilde{t}_{\xi,0}))
}\!\!\times\!\!\!
\sum_{j=0}^{1}
[\mathbf{P}_{\!k}]_{i,j}
e^{\zeta_{k}[f_{k}\!(j)-\hat{C}_{\!k}\!(\xi)]}
\Phi_{\!k}\!(j).
\end{aligned}
\end{equation}
\vspace{-0.05cm}
\noindent Hence, condition \eqref{eq:one_step_martingale_condition} holds provided that
\vspace{-0.05cm}
\begin{equation}
\sum_{j=0}^{1}
[\mathbf{P}_k]_{i,j}
e^{\zeta_k[f_{k}(j)-\hat{C}_k(\xi)]}
\Phi_k(j)
=
\Phi_k(i),
\quad i\in\{0,1\}.
\label{eq:martingale_condition}
\vspace{-0.05cm}
\end{equation} To construct the desired $\zeta_k^\star$ and $\Phi_k(\cdot)$ satisfying \eqref{eq:martingale_condition}, define the exponentially tilted transition matrix as
\begin{equation}
\vspace{-0.025cm}
\left[
\left(\mathbf{P}_k\right)^{\zeta_k}
\right]_{i,j}
=
[\mathbf{P}_k]_{i,j}
e^{\zeta_k[f_{k}(j)-\hat{C}_{k}(\xi)]},
\quad i,j\in\{0,1\}.
\label{eq:tilted_transition_matrix}
\vspace{-0.025cm}
\end{equation}
For the two-state MMOO process,
\begin{equation}
\left(\mathbf{P}_k\right)^{\zeta_k}
\!=\!
\begin{bmatrix}
(1-p_k)e^{-\zeta_k\hat{C}_k(\xi)}\!\!\!\!
&
p_k e^{\zeta_k(R_k-\hat{C}_k(\xi))}
\\
q_k e^{-\zeta_k\hat{C}_k(\xi)}\!\!\!\!
&
(1-q_k)e^{\zeta_k(R_k-\hat{C}_k(\xi))}
\end{bmatrix}\!.
\label{eq:MMOO_tilted_matrix}
\end{equation}
Since $(\mathbf{P}_k)^{\zeta_k}$ is an irreducible non-negative matrix, the Perron--Frobenius theorem guarantees that $sp((\mathbf{P}_k)^{\zeta_k})$ is a simple eigenvalue with a strictly positive right eigenvector. We choose $\zeta_k^\star>0$ as the unique positive solution to $sp((\mathbf{P}_k)^{\zeta_k})\!=\!1$, which is equivalent to $\det\!\left((\mathbf{P}_k)^{\zeta_k}-\mathbf{I}\right)\!=\!0,$ i.e.,
\begin{equation}
\vspace{0.05cm}
\zeta_k^\star
=
\left\{
\zeta_k>0:
\begin{aligned}
&
\left[
(1-p_k)e^{-\zeta_k \hat{C}_{k}(\xi)}-1
\right]\times
\\
&
\left[
(1-q_k)e^{\zeta_k(R_k-\hat{C}_{k}(\xi))}-1
\right]
\\
&
-p_k q_k
e^{-\zeta_k \hat{C}_{k}(\xi)}
e^{\zeta_k(R_k-\hat{C}_{k}(\xi))}
\!=\!0
\end{aligned}
\right\}\!.
\end{equation}
\vspace{0.05cm}
Let $\boldsymbol{\Phi}_k=[\Phi_k(0),\Phi_k(1)]^T$ be the corresponding positive right eigenvector, i.e.,$(\mathbf{P}_k)^{\zeta_k^\star}\boldsymbol{\Phi}_k=\boldsymbol{\Phi}_k .$
Then, based on \eqref{eq:tilted_transition_matrix}, this eigenvector equation is exactly \eqref{eq:martingale_condition} with $\zeta_k=\zeta_k^\star$.
Substituting this $\zeta_k^\star$ and the corresponding $\Phi_k(\cdot)$ into \eqref{eq:backlog_martingale}, the one-step condition \eqref{eq:one_step_martingale_condition} is satisfied. Therefore, $M_k^L(\tilde{t}_{\xi,r})$ is a martingale. This completes the
proof.
\end{proof}

Based on Proposition~\ref{prop:backlog_martingale}, we next derive an upper bound on the {delay-outage} probability. According to the remaining delay budget $D_k^{\mathrm{rem}}(\xi)$ defined in \eqref{eq:remaining_delay_budget}, a delay outage occurs if the DT-phase queueing delay exceeds it, i.e.,
\begin{equation}
D_k(\tilde{t}_{\xi,r})>D_k^{\max}
\Longleftrightarrow
D_k^{\mathrm{queue}}(\tilde{t}_{\xi,r})
>
D_k^{\mathrm{rem}}(\xi).
\end{equation}
\noindent Define the stopping time as the first time when the net-input
process crosses the threshold $D_k^{\mathrm{rem}}(\xi)\hat{C}_k(\xi)$, i.e.,
\begin{equation}
\vspace{0.05cm}
r_k^{\mathrm{stop}}
\!\!\!=\!\!
\inf\!\left\{\!
r\!>\!0\!:\!
{B_k\!(\xi)}
\!\!+\!\!
A_k\!(\tilde{t}_{\xi,0},\tilde{t}_{\xi,r})
\!\!-\!
r\hat{C}_{\!k}\!(\xi)
\!\!\geq\!\!
D_k^{\mathrm{rem}}\!(\xi)\hat{C}_{\!k}\!(\xi)
\!\right\}\!.
\label{eq:stopping_time_delay}
\end{equation}
\vspace{0.05cm}
According to the DT-phase queueing delay definition in \eqref{eq:queueing-delay definition}, the event $D_k^{\mathrm{queue}}(\tilde{t}_{\xi,r})>D_k^{\mathrm{rem}}(\xi)$ means that the DT-phase arrivals, including the initial service debt $B_k(\xi)$, cannot be cleared within the remaining delay budget. This yields the net-input process overflow event in equality (a) in \eqref{eq:delay_event_to_queue_pressure}, as shown at the bottom of next page. The inequality in (b) is due to the enlarged event space by considering all admissible starting slots $r'$ over $0\leq r'\leq r-D_k^{\mathrm{rem}}(\xi)$. Substituting
$S_k(\tilde{t}_{\xi,r'},\tilde{t}_{\xi,r})
=(r-r')\hat{C}_k(\xi)$ gives (c), and the change of variable
$u=r-r'$ leads to (d) , where $D_k^{\mathrm{rem}}(\xi)\leq u\leq r$. In (e), stationarity of the arrival
process allows the interval of length
$u-D_k^{\mathrm{rem}}(\xi)$ to be shifted to
$[\tilde{t}_{\xi,D_k^{\mathrm{rem}}(\xi)},\tilde{t}_{\xi,u})$;
extending the supremum from
$D_k^{\mathrm{rem}}(\xi)\leq u\leq r$ to
$u\geq D_k^{\mathrm{rem}}(\xi)$ then yields the upper bound.
Equality (f) follows by decomposing
$u\hat{C}_k(\xi)$ into
$(u-D_k^{\mathrm{rem}}(\xi))\hat{C}_k(\xi)
+D_k^{\mathrm{rem}}(\xi)\hat{C}_k(\xi)$, thereby isolating the
crossing threshold
$D_k^{\mathrm{rem}}(\xi)\hat{C}_k(\xi)$.
Setting $v=u-D_k^{\mathrm{rem}}(\xi)$ gives (g), while resetting
the time origin and relabeling the dummy variable $v$ as $r$
gives (h).

\refstepcounter{equation}
\label{eqnum:delay_event_to_queue_pressure}
\refstepcounter{equation}
\label{eqnum:stopping_inequality_delay} 

\begin{theorem}[ {Delay-Outage} Probability Bound]\label{thm:delay_outage_probability_bound}
For user $k$ in block $\xi$, under the conditions of Proposition~\ref{prop:backlog_martingale}, the {delay-outage} probability satisfies
\begin{equation}
\vspace{-0.1cm}
\Pr\!\left\{\!
D_k(\tilde{t}_{\xi,r})\!\!>\!\!D_k^{\max}\!
\right\}
\!\!\leq\!\!
\frac{\Phi_{k,\max}}{\Phi_{k,\min}}
e^{\!
-\zeta_k^\star
\left[
D_k^{\mathrm{rem}}(\xi)\hat{C}_k(\xi)
-\!
{B_k(\xi)}
\right]},
\label{eq:delay_violation_bound}
\end{equation}
where
\begin{equation}
\vspace{-0.025cm}
\Phi_{k,\min}=\min_{z\in\{0,1\}}\{\Phi_k(z)\}, \Phi_{k,\max}=\max_{z\in\{0,1\}}\{\Phi_k(z)\}.
\end{equation}
\end{theorem}
\vspace{-0.05cm}
\begin{proof}
From \eqref{eq:delay_event_to_queue_pressure}, the DT-phase queueing {delay-outage} is upper bounded by the first crossing event of the net-input
process. Hence,
\begin{equation}
\vspace{-0.05cm}
\Pr\!\left\{
D_k^{\mathrm{queue}}(\tilde{t}_{\xi,r})
>
D_k^{\mathrm{rem}}(\xi)
\right\}
\leq
\Pr\!\left\{
r_k^{\mathrm{stop}}<\infty
\right\}.
\label{eq:delay_to_stopping_event}
\end{equation} Applying the optional stopping theorem
to the
martingale stopped at $r_k^{\mathrm{stop}}\wedge r$ yields
\begin{equation}
\begin{aligned}
&\mathbb{E}\!\left[M_k^L(\tilde{t}_{\xi,0})\right]
=\mathbb{E}\!\left[M_k^L\!\!\left(\tilde{t}_{\xi,r_k^{\mathrm{stop}}\wedge r}\right)\right]
= e^{\zeta_k^\star {B_k(\xi)}}\\
&\!\!=\!\mathbb{E}\!\left[M_k^L\!\!\left(\tilde{t}_{\xi,r_k^{\mathrm{stop}}}\!\right)\!\mathbb{I}_{\{r_k^{\mathrm{stop}}\!<r\}}\right]
\!+\!\mathbb{E}\!\left[M_k^L\!\!\left(\tilde{t}_{\xi,r}\right)\!\mathbb{I}_{\{r_k^{\mathrm{stop}}\geq r\}}\right]\\
&\!\!\geq\!
\mathbb{E}\!\left[M_k^L\!\!\left(\tilde{t}_{\xi,r_k^{\mathrm{stop}}}\!\right)\!\mathbb{I}_{\{r_k^{\mathrm{stop}}\!<r\}}\right]\\
&\!\!\geq\!
\Pr\!\left\{r_k^{\mathrm{stop}}\!\!<\!\!r\right\}
\frac{\Phi_{k,\min}}
{\Phi_k(Z_k(\tilde{t}_{\xi,0}))}
e^{\zeta_k^{\star}D_k^{\mathrm{rem}}(\xi)\hat{C}_k(\xi)}\!.
\end{aligned}
\label{eq:stopping_inequality_delay}
\end{equation} Rearranging it yields
\begin{equation}
\Pr\!\left\{
r_k^{\mathrm{stop}}\!<\!r
\right\}
\!\! \leq \!\!
\frac{\Phi_k(Z_k(\tilde{t}_{\xi,0}))}{\Phi_{k,\min}}
e^{
-\zeta_k^\star
\left[
D_k^{\mathrm{rem}}(\xi)\hat{C}_k(\xi)
-
{B_k(\xi)}
\right]}.
\label{eq:finite_stopping_bound}
\end{equation}
Letting $r\to\infty$ and using
$\Phi_k(Z_k(\tilde{t}_{\xi,0}))\leq\Phi_{k,\max}$, we have
\begin{equation}
\Pr\!\left\{
r_k^{\mathrm{stop}}<\infty
\right\}
\leq 
\frac{\Phi_{k,\max}}{\Phi_{k,\min}}
e^{
-\zeta_k^\star
\left[
D_k^{\mathrm{rem}}(\xi)\hat{C}_k(\xi)
-
{B_k(\xi)}
\right]}.
\label{eq:infinite_stopping_bound}
\end{equation}
Since $D_k(\tilde{t}_{\xi,r})>D_k^{\max}$ is equivalent to
$D_k^{\mathrm{queue}}(\tilde{t}_{\xi,r})>D_k^{\mathrm{rem}}(\xi)$, combining
\eqref{eq:delay_to_stopping_event} and \eqref{eq:infinite_stopping_bound} completes the proof.
\end{proof}

Note that although \eqref{eq:delay_violation_bound} provides an upper bound on the delay outage probability, it is dependent on $B_k(\xi)$, i.e., the
actual backlog at the beginning of the DT phase, which is random and unknown at the moment of MA position optimization. To tackle this challenge, we further develop a deterministic sufficient condition that ensures \eqref{eq:delay_violation_bound}.

\begin{lemma}\label{lem:high_probability_backlog_bound}
For user $k$ in block $\xi$, define an auxiliary quantity $\delta_k$ with $0<\delta_k<\varepsilon_k$, and
\begin{equation}
\widehat{A}_{k}(\xi)
\!=\!
\inf_{\theta>0}\!
\frac{
\ln\!
\left(\!
\mathbf{e}_{Z_k(t_{\xi,0})}^{T}\!
\!\left[
\mathbf{P}_k
\mathrm{diag}(1,e^{\theta R_k})
\right]^{{\!N_{\mathrm{mov}}}(\xi)}\!
\mathbf{1}
\right)
\!\!+\!
\ln(\frac{1}{\delta_k})
}
{\theta}.
\label{eq:chernoff_arrival_quantile_bound}
\end{equation} Let $\widehat{B}_{k}(\xi)
\triangleq
Q_k(t_{\xi,0})
+
\widehat{A}_{k}(\xi)
\label{eq:high_probability_backlog_definition}$. Then, the actual backlog at the beginning of the DT phase satisfies
\begin{equation}
\Pr
\left\{
B_k(\xi)>\widehat{B}_{k}(\xi)
\right\}
\leq
\delta_k .
\label{eq:backlog_exceedance_probability}
\end{equation}
\end{lemma}

\begin{proof}
According to \eqref{eq:start_backlog}, the backlog at the beginning of the DT
phase is $B_k(\xi)
=
Q_k(t_{\xi,0})
+
A_k(t_{\xi,0},\tilde{t}_{\xi,0}).$ For any $x>0$ and $\theta>0$, conditioning on the initial traffic state
$Z_k(t_{\xi,0})$ and applying the Chernoff bound yield
\begin{align}
&\Pr\!\left\{\!
A_k(t_{\xi,0},\tilde{t}_{\xi,0})\!\!>\!\!x
\middle|
Z_{k}(t_{\xi,0})\!
\right\}
\!\!\leq\!
e^{\!-\theta x}
\mathbb{E}\!\left[
e^{\theta \!A_k(t_{\xi,0},\tilde{t}_{\xi,0}\!)}
\middle|
Z_k(t_{\xi,0})
\!\right]
\nonumber\\
&=
e^{-\theta x}
\mathbf{e}_{Z_k(t_{\xi,0})}^{T}
\left[
\mathbf{P}_k
\mathrm{diag}(1,e^{\theta R_k})
\right]^{N_{\mathrm{mov}}(\xi)}
\mathbf{1}.
\label{eq:Chernoff movement arrivals}
\end{align}
If the right-hand side of
\eqref{eq:Chernoff movement arrivals} is no larger than
$\delta_k$, it is sufficient that
\begin{equation}
x\!\geq\!\frac{\ln\!\left(\!\mathbf{e}_{Z_k(t_{\xi,0})}^T\!\left[\mathbf{P}_k\mathrm{diag}(1,e^{\theta \!R_k})\right]^{\!N_\mathrm{mov}(\xi)}\!\mathbf{1}\!\right)\!+\!\ln(\frac{1}{\delta_k})}{\theta}.
\end{equation} Since \eqref{eq:Chernoff movement arrivals} holds for every
$\theta>0$, the tightest threshold of $x$ provided by this family of Chernoff bounds is \eqref{eq:chernoff_arrival_quantile_bound}. Consequently, by replacing $x$ in \eqref{eq:Chernoff movement arrivals}  with \eqref{eq:chernoff_arrival_quantile_bound}, we have
\begin{equation}
\Pr\!\left\{
A_k(t_{\xi,0},\tilde{t}_{\xi,0})
>
\widehat{A}_k(\xi)
\,\middle|\,
Z_k(t_{\xi,0})
\right\}
\leq \delta_k.
\label{eq:arrival_exceedance_probability}
\end{equation} Since \eqref{eq:arrival_exceedance_probability} holds for every realization of the initial traffic state $Z_k(t_{\xi,0})$, the law of total probability gives
\begin{align}
&\Pr\!\left\{
A_k(t_{\xi,0},\tilde{t}_{\xi,0})>\widehat{A}_k(\xi)
\right\}
\nonumber\\
&=\!\!\!\!\!\!
\sum_{z\in\{0,1\}}\!\!\!\!\!
\Pr\!\left\{\!
A_k(t_{\xi,0},\tilde{t}_{\xi,0})\!\!>\!\!\widehat{A}_k(\xi)
\middle|
Z_k(t_{\xi,0})\!=\!z\!
\right\}\!
\Pr\!\left\{Z_k(t_{\xi,0})\!=\!z\right\}
\nonumber\\
&\le
\delta_k\!\!\!\!\!
\sum_{z\in\{0,1\}}\!\!\!\!\!
\Pr\!\left\{Z_k(t_{\xi,0})\!=\!z\right\}
=\delta_k{.}
\label{eq:unconditional_arrival_exceedance}
\end{align}
Since
$
B_k(\xi)\!=\!Q_k(t_{\xi,0})\!+\!A_k(t_{\xi,0},\tilde{t}_{\xi,0})
$
and
$
\widehat{B}_k(\xi)\!=\!Q_k(t_{\xi,0})\!+\!\widehat{A}_k(\xi),
$
we obtain \eqref{eq:backlog_exceedance_probability}. This completes the proof.
\end{proof}

\begin{proposition}
[Deterministic Sufficient Condition]
\label{prop:deterministic_delay_constraint} If the following deterministic condition holds:
\begin{equation}
\frac{\Phi_{k,\max}}{\Phi_{k,\min}}
e^{-\zeta_k^\star
\left[
D_k^{\mathrm{rem}}(\xi)\hat{C}_k(\xi)
-\widehat{B}_{k}(\xi)
\right]}
\leq
\varepsilon_k-\delta_k ,
\label{eq:deterministic_delay_condition}
\end{equation}
then the delay-outage constraint 
$\Pr\!
\left\{
D_k(\tilde{t}_{\xi,r})\!>\!D_k^{\max}
\right\}
\!\!\leq\!
\varepsilon_k
\label{eq:delay_outage_constraint_original}$
is guaranteed.
\end{proposition}

\begin{proof}
Note that the delay-outage probability can be decomposed as
\begin{equation}
\begin{split}
&\Pr
\left\{
D_k(\tilde{t}_{\xi,r})>D_k^{\max}
\right\}
\\
&\!\leq \!
\Pr\!
\left\{\!
B_k(\xi)\!\!>\!\!\widehat{B}_k(\xi)\!
\right\}
\!+\!
\Pr\!
\left\{\!
D_k(\tilde{t}_{\xi,r})\!\!>\!D_k^{\max},
B_k(\xi)\!\! \leq \! \widehat{B}_k(\xi)\!
\right\}\!.
\end{split}
\label{eq:delay_probability_decomposition}
\end{equation}
Furthermore, the delay-outage upper bound in
\eqref{eq:delay_violation_bound}
is monotonically increasing with
respect to $B_k(\xi)$. Hence, under the condition that $B_k(\xi) \leq \widehat{B}_k(\xi)$, we have
\begin{equation}
\begin{aligned}
&
\Pr
\left\{
D_k(\tilde{t}_{\xi,r})\!>\!D_k^{\max},
B_k(\xi)\!\leq\!\widehat{B}_{k}(\xi)
\right\}
\\
&\leq
\frac{\Phi_{k,\max}}{\Phi_{k,\min}}
e^{-\zeta_k^\star
\left[
D_k^{\mathrm{rem}}(\xi)\hat{C}_k(\xi)
-\widehat{B}_{k}(\xi)
\right]},
\end{aligned}
\label{eq:conditional_delay_outage_bound}
\end{equation} where the inequality is obtained by replacing $B_k(\xi)$ in \eqref{eq:delay_violation_bound} with $\widehat{B} _k(\xi)$. Substituting
\eqref{eq:conditional_delay_outage_bound}
and \eqref{eq:backlog_exceedance_probability}
into
\eqref{eq:delay_probability_decomposition},
we obtain
\vspace{-0.05cm}
\begin{equation}
\begin{split}
&
\Pr\!
\left\{
D_k(\tilde{t}_{\xi,r})\!>\!D_k^{\max}
\right\}\!\leq\!
\delta_k
\!+\!
\frac{\Phi_{k,\max}}{\Phi_{k,\min}}
e^{-\zeta_k^\star
\left[
D_k^{\mathrm{rem}}(\xi)\hat{C}_k(\xi)
-\widehat{B}_{k}(\xi)
\right]}
\\
&\leq
\delta_k+(\varepsilon_k-\delta_k)
=
\varepsilon_k ,
\end{split}
\raisetag{12pt}
\vspace{-0.25cm}
\end{equation}
where the last inequality follows from
\eqref{eq:deterministic_delay_condition}.
This completes the proof.
\end{proof}
Based on Proposition~\ref{prop:deterministic_delay_constraint}, we derive a sufficient condition, i.e., \eqref{eq:deterministic_delay_condition}{,} for the delay-outage constraints in \eqref{eq:P1_delay}.

\vspace{0.15cm}
\subsection{Delay-Guaranteed Service Rate (DGSR)}
To facilitate the system design, we further
derive an equivalent minimum service rate for each user $k$ in block $\xi$ based on the deterministic sufficient condition in \eqref{eq:deterministic_delay_condition} which is referred
to as the DGSR and denoted as $\widehat{C}_k^{\mathrm{DGSR}}(\xi)$.

Taking logarithms on both sides of \eqref{eq:deterministic_delay_condition}, we obtain
\begin{equation}
\hat{C}_k(\xi)
\!\geq\!
\frac{1}{D_k^{\mathrm{rem}}(\xi)}\!
\left[
\frac{1}{\zeta_k^\star}
\ln\!
\left(
\frac{\Phi_{k,\max}}
{\Phi_{k,\min}(\varepsilon_k\!-\!\delta_k)}
\right)
\!\!+\!
{\widehat{B}_{k}(\xi)}
\right].
\label{eq:dgsr_rate_condition}
\end{equation}
Accordingly, we can set
\begin{equation}
\hat{C}_k^\mathrm{DGSR}(\xi)\!\!=\!\!
\frac{1}{D_k^{\mathrm{rem}}(\xi)}\!
\left[\!
\frac{1}{\zeta_k^\star}
\ln\!\!
\left(\!
\frac{\Phi_{k,\max}}
{\Phi_{k,\min}(\varepsilon_k\!-\!\delta_k)}
\!\right)
\!\!+\!\!
{\widehat{B}_{k}(\xi)}
\!\right]\!\!.
\label{eq:dgsr_closed_form}
\end{equation}

The deterministic sufficient condition for the delay-outage constraint in \eqref{eq:deterministic_delay_condition} can be equivalently recast as the following service rate constraints, i.e.,
\begin{equation}
\hat{C}_k(\xi)
\geq
\hat{C}_{k}^{\mathrm{DGSR}}(\xi),
\quad k\in\mathcal{K}.
\label{eq:dgsr_service_constraint}
\end{equation}

To obtain a tractable formulation for ~\probref{prob:P1}, we replace the original probabilistic delay-outage constraint \eqref{eq:P1_delay} with
\eqref{eq:dgsr_service_constraint}. This leads to the following
optimization problem:
\begin{subequations}
\begin{align}
\optprob{prob:P2}:\!\!
\min_{\{\boldsymbol{\Lambda}(\xi),\boldsymbol{\Omega}(\xi),\boldsymbol{\omega}(\xi)\}}
&
\lim_{|\boldsymbol{\Xi}|\rightarrow\infty}\frac{1}{|\boldsymbol{\Xi}|}
\sum_{\xi\in\boldsymbol{\Xi}}
\mathrm{Tr}\!\left(\boldsymbol{\omega}(\xi)\boldsymbol{\omega}^{H}(\xi)\right) \label{eq:P2_obj}\\
\mathrm{s.t.}\quad
&\hat{C}_k(\xi)
\geq
\hat{C}_{k}^{\mathrm{DGSR}}(\xi), \forall k,\xi, \label{eq:P2_service}\\
&\eqref{eq:P1_Tx_region}-\eqref{eq:P1_queue}. \label{eq:P2_inherited}
\end{align}
\end{subequations}Although \probref{prob:P2} is a deterministic problem, it remains challenging to solve optimally due to inter-block coupling of MA positions and queue states,
as well as intra-block nonconvex coupling among MA positions, beamforming vectors, and  service rates.

\vspace{0.25cm}
\section{Proposed Solution to \probref{prob:P2}}
\label{sec:proposed_solution}
To tackle \probref{prob:P2}, we propose {an AO algorithm with a two-layer framework}. First, the long-horizon problem with $|\boldsymbol{\Xi}|\rightarrow\infty$ is approximated by a sequence of short-horizon problems. Each short-horizon problem is then reformulated as an outer movement-budget optimization and an inner joint optimization of the MA destination positions and transmit beamforming. After solving the current short-horizon problem, the optimized system states are only implemented in the first block, and the horizon slides forward.

\subsection{Short-Horizon Approximation}
\label{sec:short_horizon}
At the beginning of block $\xi$, the current system state is denoted as
\begin{equation}
\vspace{-0.05cm}
\mathcal S(\xi)=\{\boldsymbol\Lambda(t_{\xi,0}),\boldsymbol\Omega(t_{\xi,0}),\mathbf Q(t_{\xi,0})\},
\label{eq:DORA_state}
\end{equation}
where $\mathbf Q(t_{\xi,0})=[Q_1(t_{\xi,0}),\ldots,Q_K(t_{\xi,0})]^T$.
Let $H$ denote the 
\par\noindent length of the short horizon for approximation. Accordingly, the optimization window is given by the set $\mathcal H_\xi=\{\xi,\ldots,\min\{\xi+H-1,|\boldsymbol\Xi|-1\}\}$.
At each block $\xi$, the original long-horizon problem \probref{prob:P2} is approximated by the following finite-window optimization problem over $\mathcal H_\xi$:
\begin{subequations}\label{eq:short_horizon_problem}
\begin{align}
&\optprob{prob:P3}(\mathcal H_\xi):\!
\min_{\{\boldsymbol\Lambda(\nu),\boldsymbol\Omega(\nu),\boldsymbol\omega(\nu)\}_{\nu\in\mathcal H_\xi}}\!
\sum_{\nu\in\mathcal H_\xi}\!\mathrm{Tr}\big(\boldsymbol\omega(\nu)\boldsymbol\omega^H(\nu)\big)
\label{eq:P3_obj}\\
&\hspace{3.0cm}\mathrm{s.t.}\quad
\eqref{eq:P2_service},\ \eqref{eq:P2_inherited},\ \forall\nu\in\mathcal H_\xi.
\label{eq:P3_constraints}
\end{align}
\end{subequations} The constraints inherited from \probref{prob:P2} are imposed on each block $\nu\in\mathcal H_\xi$ by replacing $\xi$ with $\nu$, with the window initialized by $\mathcal S(\xi)$. Moreover, the factor $1/|\mathcal H_\xi|$ is omitted in the objective function. 

It is worth noting that, due to the position-continuity and queue-evolution constraints, the decisions across blocks remain temporally coupled. Therefore, \probref{prob:P3}$(\mathcal H_\xi)$ constitutes a joint finite-window optimization problem rather than a collection of independent per-block problems. Next, we develop an efficient solution to \probref{prob:P3}$(\mathcal H_\xi)$.
\vspace{0.15cm}
\subsection{Two-Layer Problem Reformulation}
\label{sec:two_layer}
For \probref{prob:P3}$(\mathcal H_\xi)$, the movement duration depends on the destination MA positions through the maximum travel distance and
slot rounding in \eqref{slot rounding}. To decouple the MA movement duration from the continuous MA position optimization, we introduce an integer movement-slot budget
$\ell(\nu)\in\{0,\ldots,N_0-1\}$ for each block
$\nu\in\mathcal{H}_{\xi}$, which specifies the maximum number of slots allocated for MA movement. Thus, the actual movement duration satisfies
$N_{\mathrm{mov}}(\nu)\leq \ell(\nu)$, and the remaining delay budget is
$D_{k,\ell}^{\mathrm{rem}}(\nu)\triangleq D_k^{\max}-\ell(\nu)$. Let $\boldsymbol{\ell}_{\xi} = \{\ell(\nu)\}_{\nu \in \mathcal{H}_{\xi}}$ denote the sequence of the movement delay budget across $\mathcal{H}_{\xi}$, which is
feasible if $D_{k,\ell}^{\mathrm{rem}}(\nu) >
0, \forall k \in \mathcal{K},\nu \in \mathcal{H}_{\xi}$. Based on this, we define
\begin{equation}
\widehat{A}_{k,\ell}(\nu)\!\!
\triangleq \!\!
\inf_{\theta>0}\!
\frac{\!
\ln\!\Big(\!
\mathbf{e}_{Z_k(t_{\nu,0})}^{T}\!\!
\left[
\mathbf{P}_k
\mathrm{diag}(1,e^{\theta R_k})
\right]^{\!\ell(\nu)}\!
\mathbf{1}\!
\Big)\!
\!+\! \ln\!\left(\!\frac{1}{\delta_k}\!\right)
}{\theta}\!\!
\label{eq:DORA_budgeted_arrival}
\end{equation} and 
$\widehat B_{k,\ell}(\nu)
\triangleq
Q_k(t_{\nu,0})
+\widehat A_{k,\ell}(\nu),$
which yield the following budget-dependent DGSR
\begin{equation}
\hat{C}_{k,\ell}^{\mathrm{DGSR}}(\nu)
\!\!=\!\!
\frac{1}{D_{k,\ell}^{\mathrm{rem}}(\nu)}\!\!
\left[\!
\frac{1}{\zeta_k^\star}\!
\ln\!\!\left(\!\!
\frac{\Phi_{k,\max}}
{\Phi_{k,\min}(\varepsilon_k\!-\!\delta_k)}\!\!
\right)
\!\!+\!\!
\widehat{B}_{k,\ell}(\nu)
\! \right]\!\!.
\label{eq:DORA_DGSR_budget}
\end{equation}

For a given sequence $\boldsymbol\ell_\xi$, let $J_\xi^\star(\boldsymbol\ell_\xi)$ denote the optimal value of \probref{prob:P3}$(\mathcal H_\xi)$. As such \probref{prob:P3}$(\mathcal H_\xi)$ can be recast as the following two-layer problem, i.e.,
\begin{subequations}\label{eq:P3_outer}\vspace{-0.05cm}
\begin{align}
&\!\!\!\!\!\!\!\!\!\!\textnormal{(\ref{prob:P3}$(\mathcal H_\xi)$-Outer):}\,\,\min_{\boldsymbol{\ell}_{\xi}}
J_{\xi}^{\star}(\boldsymbol{\ell}_{\xi})
\label{eq:P3_outer_obj}\\
&\quad\quad\quad\,\,\,\,\mathrm{s.t.}\,\,\,
\ell(\nu)\in\{0,\ldots,N_0-1\},
\forall \nu\in\mathcal H_\xi,
\label{eq:P3_outer_budget}\\
&\quad\quad\quad\,\,\,\,\hphantom{\mathrm{s.t.}\,\,\,}
D_{k,\ell}^{\mathrm{rem}}(\nu)>0,
\forall k\in\mathcal K,\ \nu\in\mathcal H_\xi,
\label{eq:P3_outer_delay}
\end{align}
\end{subequations}
and
\vspace{-0.05cm}
\begin{subequations}\label{eq:P3_inner}
\begin{align}
&\!\!\!\!\textnormal{(\ref{prob:P3}$(\mathcal H_\xi)$-Inner):}\,\,J_{\xi}^{\star}(\boldsymbol{\ell}_{\xi})
\!=\!\!\!\!\!\!
\min_{\substack{\!
\{\boldsymbol{\Lambda}(\nu),\boldsymbol{\Omega}(\nu),\\
\boldsymbol{\omega}(\nu)\}_{\nu\in\mathcal H_{\xi}}
\!}}\!\!\!\!
\sum_{\,\,\,\,\,\nu\in\mathcal H_{\xi}}\!\!\!\!\!
\mathrm{Tr}\!\left(
\boldsymbol{\omega}(\nu)\boldsymbol{\omega}^{H}(\nu)
\right)\!,
\label{eq:P3_inner_obj}\\
&\quad\quad\quad\quad\,\,\mathrm{s.t.}\,\,\,
\hat C_k(\nu)\! \geq \!
\hat C_{k,\ell}^{\mathrm{DGSR}}(\nu),\forall k\in\mathcal K,\!\ \nu\in\mathcal H_\xi,
\label{eq:P3_inner_service}\\
&\hphantom{\mathrm{s.t.}\quad\quad\quad\quad\,\,\,\,\,}
N_{\mathrm{mov}}(\nu)\leq\ell(\nu),
\forall \nu\in\mathcal H_\xi,
\label{eq:P3_inner_movement}\\
&\hphantom{\mathrm{s.t.}\quad\quad\quad\quad\,\,\,\,\,}
\eqref{eq:P1_Tx_region}\text{--}\eqref{eq:P1_queue},
\forall\nu\in\mathcal H_\xi.
\label{eq:P3_inner_inherited}
\end{align}
\end{subequations}

\vspace{0.05cm}
\subsection{Proposed Solution to \textup{(}\uprightref{prob:P3}$(\mathcal H_\xi)$-\textup{Inner)}}
\label{sec:inner_solution}
\vspace{-0.05cm}
We employ AO to handle the coupling between the destination MA positions and transmit beamforming in (\ref{prob:P3}$(\mathcal H_\xi)$-Inner) by decomposing it into the following subproblems. 

\subsubsection{Transmit Beamforming Optimization with Fixed Destination MA Positions}
\label{sec:beamforming_update} For fixed destination MA positions, it can be shown that the constraints in \eqref{eq:P3_inner_service} can be recast as
\begin{equation}
\gamma_k(\nu)\geq\Gamma_{k,\ell}(\nu),\quad\forall k,\nu\in\mathcal{H}_\xi.
\label{SINR_threshold}
\end{equation} where $\Gamma_{k,\ell}(\nu)$ denotes the SINR threshold induced by the DGSR. By applying the standard phase-rotation technique, i.e., 
\begin{equation}
\Im\{\mathbf{h}_k^H(\nu)\boldsymbol{\omega}_k(\nu)\}=0,
\quad
\Re\{\mathbf{h}_k^H(\nu)\boldsymbol{\omega}_k(\nu)\}\geq0,
\label{eq:DORA_phase_rotation}
\end{equation}
to \eqref{SINR_threshold}, it can be further simplified as a second-order cone constraint:
\begin{equation}
\left\|\!
\left[
\{
\mathbf h_k^H\!(\nu)\boldsymbol{\omega}_j(\nu)
\}_{j\neq k},
\sigma_k
\right]^{\!T}\!
\right\|_2
\!\leq\!
\frac{
\Re
\{
\mathbf h_k^H\!(\nu)
\boldsymbol{\omega}_k(\nu)
\}
}
{\sqrt{\Gamma_{k,\ell}(\nu)}} ,\forall k,\nu.
\label{eq:second-order cone constraint}
\end{equation}
Hence, the beamforming subproblem is formulated as the following second-order cone programming (SOCP) problem:
\begin{subequations}
\begin{align}
\optprob{prob:P4}:\!\!
\min_{\{\boldsymbol{\omega}(\nu)\}}
&
\sum_{\nu\in\mathcal{H}_{\xi}}
\mathrm{Tr}\!\left(
\boldsymbol{\omega}(\nu)\boldsymbol{\omega}^{H}(\nu)
\right)
\label{eq:P4_obj}\\
\mathrm{s.t.}\quad
& ~\eqref{eq:DORA_phase_rotation},~\eqref{eq:second-order cone constraint}.
\end{align}
\end{subequations}
Since \probref{prob:P4} is convex, it can be optimally solved by invoking the interior-point algorithm.

\subsubsection{Destination MA Position Optimization for any Given Transmit Beamforming}
\label{sec:position_update} For fixed beamforming vectors, the second inner-layer subproblem optimizes the MA positions. 

The MA positions affect both the desired signal gains and interference
terms, resulting in a nonconvex coupling.
Therefore, the successive convex approximation (SCA) method is employed
to obtain a tractable solution.
 Let $\mathbf{x}(\nu)
=
[
\boldsymbol{\Lambda}^{T}(\nu),
\boldsymbol{\Omega}^{T}(\nu)
]^T$ denote the collection of all MA positions. At the $\kappa$-th SCA iteration, a concave lower bound of the rate is constructed as
\begin{equation}
\begin{aligned}
\widetilde{C}_{k}^{(\kappa)}(\mathbf{x}(\nu))
&=
C_k(\mathbf{x}^{(\kappa)}(\nu))
+
\nabla_{\mathbf{x}}C_k(\mathbf{x}^{(\kappa)}(\nu))^{T}
(\mathbf{x}(\nu)\\
&\quad
-\mathbf{x}^{(\kappa)}(\nu))-
\frac{\rho_k}{2}
\|\mathbf{x}(\nu)-\mathbf{x}^{(\kappa)}(\nu)\|_2^2,
\end{aligned}
\label{eq:DORA_rate_lower_bound}
\end{equation}
where $\rho_k\!\!>\!\!0$ is no smaller than the local Lipschitz constant of $\nabla_{\!\mathbf{x}}C_k(\mathbf{x})$. Accordingly, the nonconvex rate constraint is conservatively replaced by
\begin{equation}
{\tau \widetilde{C}_{k}^{(\kappa)}\!(\mathbf{x}(\nu))}/
{L_{\mathrm{pkt}}}
\!\geq\!
\hat{C}_{k,\ell}^{\mathrm{DGSR}}(\nu),\forall k,\nu.
\label{eq:DORA_SCA_rate_constraint}
\end{equation}
The given movement budgets impose
\begin{subequations}
\begin{align}
\left\|
\boldsymbol{\Lambda}_m\!(t_{\nu,{n_{\mathrm{s},m}}})
\!-\!
\boldsymbol{\Lambda}_m\!(t_{\nu,0})
\right\|_2
&\leq
v^{\mathrm{s}}\tau \ell(\nu),
\quad \forall m,\nu,
\label{eq:DORA_Tx_movement}\\
\left\|
\boldsymbol{\Omega}_k(t_{\nu,{n_{\mathrm{r},k}}})
-
\boldsymbol{\Omega}_k(t_{\nu,0})
\right\|_2
&\leq 
v^{\mathrm{r}}\tau \ell(\nu),
\quad \forall k,\nu.
\label{eq:DORA_Rx_movement}
\end{align}
\end{subequations}
For the spacing constraint, define
$\mathbf{d}_{ij}^{(\kappa)}\!(\nu)\!=\!
\boldsymbol{\Lambda}_i^{\!(\kappa)}(t_{\nu,{n_{\mathrm{s},i}}})
-\!
\boldsymbol{\Lambda}_j^{\!(\kappa)}(t_{\nu,{n_{\mathrm{s},j}}}).$
Its first-order lower bound is
\begin{equation}
2(\mathbf{d}_{ij}^{(\kappa)}(\nu))^{\!T}\!\!
\left[
\boldsymbol{\Lambda}_i(t_{\nu,{n_{\mathrm{s},i}}})
\!-\!\!
\boldsymbol{\Lambda}_j(t_{\nu,{n_{\mathrm{s},j}}})
\right]
\!-\!
\|\mathbf{d}_{ij}^{(\kappa)}(\nu)\|_2^2
\!\geq\!\!
D^2\!.
\label{eq:DORA_spacing_linearization}
\end{equation}

Note that, with fixed transmit beamforming, the transmit p-

\par\noindent ower objective is independent of the MA-position variables and thus cannot guide the MA-position update. Therefore, a proximal regularization objective is adopted to stabilize the SCA iterations and suppress unnecessary position oscillations within the feasible region. Combining the above approximations yields the following convex MA-position update:
\begin{subequations}
\begin{align}
\optprob{prob:P5}:\!\!
\min_{\{\boldsymbol{\Lambda}(\nu),\boldsymbol{\Omega}(\nu)\}}
&
\eta_{\Lambda}\!\!\!\!
\sum_{\nu\in\mathcal{H}_{\xi}}\!\sum_{m\in\mathcal{M}}\!
\left\|
\boldsymbol{\Lambda}_m(t_{\nu,{n_{\mathrm{s},m}}})
\!-\!
\boldsymbol{\Lambda}_m^{(\kappa)}(t_{\nu,{n_{\mathrm{s},m}}})
\right\|_2^2
\nonumber\\
&+\!
\eta_{\Omega}\!\!\!
\sum_{\nu\in\mathcal{H}_{\xi}}\!\sum_{k\in\mathcal{K}}\!
\left\|
\boldsymbol{\Omega}_k(t_{\nu,{n_{\mathrm{r},k}}})
\!-\!
\boldsymbol{\Omega}_k^{(\kappa)}(t_{\nu,{n_{\mathrm{r},k}}})
\right\|_2^2
\label{eq:P5_obj}\\
\mathrm{s.t.}\quad
&
\eqref{eq:DORA_Tx_movement},~\eqref{eq:DORA_Rx_movement},~\eqref{eq:DORA_spacing_linearization}
\label{eq:P5_movement},\\
&
\boldsymbol{\Lambda}_m(t_{\nu,{n_{\mathrm{s},m}}})
\in
\mathcal{C}^{\mathrm{s}},
\quad \forall m,\nu,
\label{eq:P5_Tx_region}\\
&
\boldsymbol{\Omega}_k(t_{\nu,{n_{\mathrm{r},k}}})
\in
\mathcal{C}^{\mathrm{r}}_k,
\quad \forall k,\nu,
\label{eq:P5_Rx_region}\\
&
{\tau \widetilde{C}_{k}\!^{\!(\kappa)}\!(\mathbf{x}(\nu))}\!/
\!{L_{\mathrm{pkt}}}
\!\!\geq\!\!
\hat{C}_{k,\ell}^{\mathrm{DGSR}}(\nu), \!\forall k,\nu,
\label{eq:P5_rate}\\
&
\boldsymbol{\Lambda}_m(t_{\nu+1,0})
=
\boldsymbol{\Lambda}_m(t_{\nu,{n_{\mathrm{s},m}}}), \forall m,\nu,
\label{eq:P5_Tx_continuity}\\
&
\boldsymbol{\Omega}_k(t_{\nu+1,0})
=
\boldsymbol{\Omega}_k(t_{\nu,{n_{\mathrm{r},k}}}), \forall k,\nu.
\label{eq:P5_Rx_continuity}
\end{align}
\end{subequations}
where $\eta_{\Lambda}\!\!>\!\!0$ and $\eta_{\Omega}\!\!>\!\!0$ are proximal regularization parameters.

\begin{algorithm}[t]
\caption{Proposed Algorithm to Solve \probref{prob:P2}}
\label{alg:DORA}
\begin{algorithmic}[1]
\STATE \textbf{Input:} Initial MA positions $\boldsymbol{\Lambda}(t_{0,0})$, $\boldsymbol{\Omega}(t_{0,0})$, queue state $\mathbf{Q}(t_{0,0})$, rolling window length $H$, candidate movement slot budgets, convergence tolerance $\epsilon_{\mathrm{conv}}$, and maximum iteration number $I_{\max}$.
\STATE Initialize $\xi=0$.
\WHILE{$\xi\leq|\boldsymbol{\Xi}|-1$}
    \STATE Observe $\mathcal{S}(\xi)$, construct $\mathcal{H}_{\xi}$, and predict the arrivals and queue states.
    \STATE Construct candidate budget sequences $\{\ell(\nu)\}_{\nu\in\mathcal{H}_{\xi}}$.
    \FOR{each candidate $\{\ell(\nu)\}_{\nu\in\mathcal{H}_{\xi}}$}
        \STATE Compute $\widehat{A}_{k,\ell}(\nu)$, $\widehat{B}_{k,\ell}(\nu)$, and $\hat{C}_{k,\ell}^{\mathrm{DGSR}}(\nu)$ for all $k$ and $\nu$.
        \STATE Initialize $\boldsymbol{\Lambda}^{(0)}(\nu)$ and $\boldsymbol{\Omega}^{(0)}(\nu)$, and set $\kappa=0$.
        \REPEAT
            \STATE Fix $\boldsymbol{\Lambda}^{(\kappa)}(\nu)$ and $\boldsymbol{\Omega}^{(\kappa)}(\nu)$; solve the SOCP subproblem \probref{prob:P4}.
            \STATE Fix $\boldsymbol{\omega}^{(\kappa+1)}(\nu)$; solve the SCA-based MA-position subproblem \probref{prob:P5}.
            \STATE Set $\kappa\leftarrow\kappa+1$.
        \UNTIL{the objective variation is below $\epsilon_{\mathrm{conv}}$ or $\kappa\!\!\geq\!\! I_{\max}$}
        \STATE Record the feasible window objective value.
    \ENDFOR
    \STATE Select the feasible candidate with the minimum window objective and execute its first-block decision.
    \STATE Update $\boldsymbol{\Lambda}(t_{\xi+1,0})$, $\boldsymbol{\Omega}(t_{\xi+1,0})$, and $\mathbf{Q}(t_{\xi+1,0})$ according to \eqref{eq:P1_Tx_continuity}--\eqref{eq:P1_queue}.
    \STATE Set $\xi\leftarrow\xi+1$.
\ENDWHILE
\STATE \textbf{Output:} Online MA target positions and beamforming matrices.
\end{algorithmic}
\end{algorithm}

\begin{figure*}[!t]
    \vspace{-0.75cm}
    \centering
    \includegraphics[width=\textwidth]{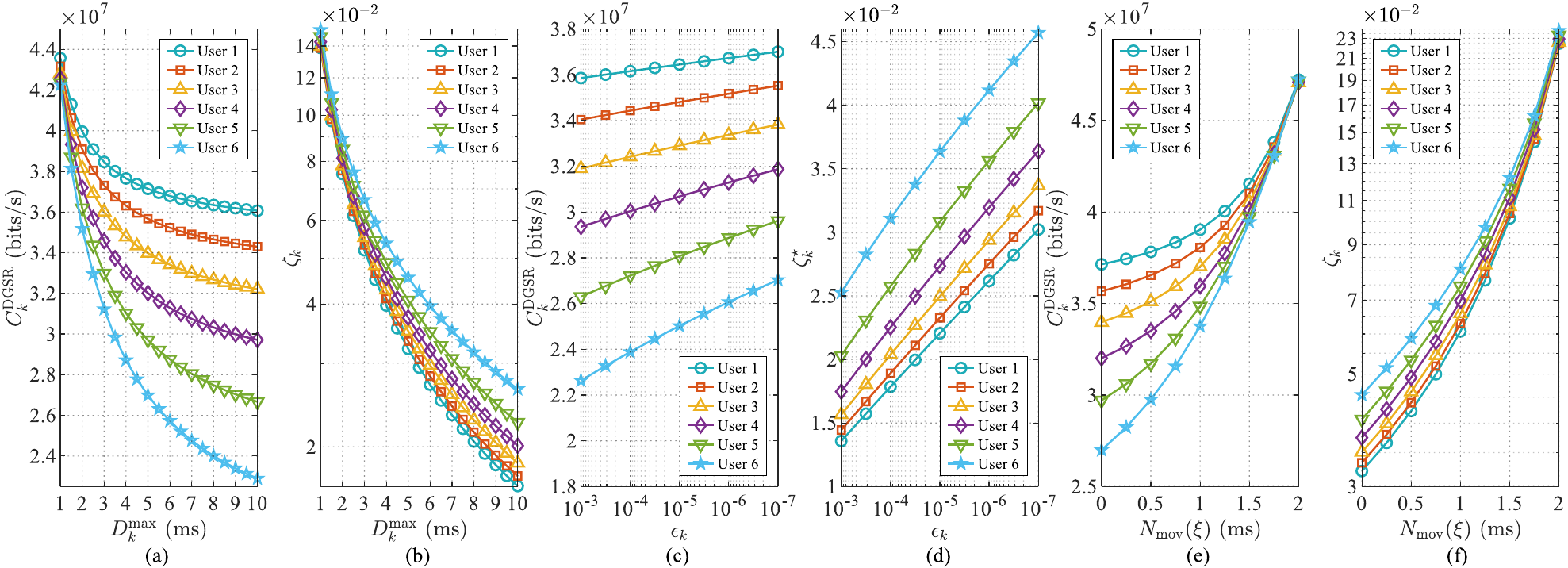}
    \vspace{-0.75cm}
    \caption{Characterization of the proposed DGSR and martingale exponent under different statistical delay parameters: (a)--(b) versus $D_k^{\max}$, (c)--(d) versus $\varepsilon_k$, and (e)--(f) versus {$N_{\rm mov}(\xi)$}.}
    \vspace{-0.15cm}
    \label{fig:DGSR_characterization}
\end{figure*}

For each movement-budget sequence, \probref{prob:P4} and \probref{prob:P5} are alternately solved. Since this yields a non-increasing objective value of (\ref{prob:P3}$(\mathcal H_\xi)$-Inner), the inner iterations converge to a stationary feasible point under standard SCA conditions. The outer layer then selects the feasible candidate with the minimum objective value. The overall complexity per rolling window is given by
\begin{equation}
    \mathcal{O}\!\left(\! \widetilde N_{\ell} I_{\max}\! \left[\! \sqrt{\!N_{\mathrm{bf}}}(HMK)^3\!\ln\!\frac{1}{\epsilon_{\mathrm{bf}}} \!\!+\!\! \sqrt{\!N_{\mathrm{ma}}}N_x^3\!\ln\!\frac{1}{\epsilon_{\mathrm{ma}}} \right] \!\right)\!, 
\end{equation} where \(\widetilde N_{\ell}\) denotes the number of feasible movement-slot budget sequences \(\boldsymbol{\ell}_{\xi}\) considered in the outer-layer search, \(I_{\max}\) is the maximum number of AO iterations, and \(N_x=2H(M+K)\). Moreover, \(N_{\rm bf}\) and \(N_{\rm ma}\) denote the numbers of constraints in (P4) and (P5), respectively, while \(\epsilon_{\rm bf}\) and \(\epsilon_{\rm ma}\) are their corresponding solution accuracies. Hence, the computational complexity grows polynomially with the prediction-window length, the number of antennas, and the number of users. After solving \probref{prob:P3}$(\mathcal{H}_{\xi})$, the optimized solutions are implemented in the current block $\xi$, and the short horizon slides forward by shifting the blocks by one. The complete procedure is summarized in Algorithm~\ref{alg:DORA}.
\vspace{-0.05cm}
\section{Numerical Results}
\label{sec:numerical_results}

\begin{table}[!t]
\centering
\vspace{-0.45cm}
\caption{Main Simulation Parameters}
\vspace{-0.15cm}
\label{tab:simulation_parameters}
\renewcommand{\arraystretch}{1.15}
\begin{tabular}{l c}
\hline
\textbf{Parameter} & \textbf{Value} \\
\hline

Carrier wavelength, $\lambda$ & $0.01$ m \\
System bandwidth, $B$ & $100$ MHz \\
Noise power, $\sigma^2$ & $-82$ dBm \\
Minimum Tx-MA spacing, $D$ & $0.5\lambda$ \\
BS-side movable region & $[-2\lambda,2\lambda]\times[-2\lambda,2\lambda]$ \\
User-specific Rx-MA movable region & $[-\lambda,\lambda]\times[-\lambda,\lambda]$ \\
Tx-MA moving speed, $v^{\rm s}$ & $0.5$ m/s \\
Rx-MA moving speed, $v^{\rm r}$ & $0.5$ m/s \\
Packet length, $L_{\rm pkt}$ & $512$ bits \\
Number of transmission blocks, $|\boldsymbol{\Xi}|$ & $2\times 10^4$ \\
ON-state arrival rate, $R_k$ & $14$ packets/slot \\
Number of channel uses, $m_{\rm s}$ & $300$ \\
Prescribed BLER, $\hbar$ & $10^{-5}$ \\
Delay violation probability, $\varepsilon_k$ & $10^{-5}$ \\
Delay threshold, $D_k^{\max}$ & $1$--$10$ ms \\
Probability budget, $\delta_k$ & $10^{-6}$ \\
Number of Monte Carlo realizations, $N_{\rm MC}$ & 50 \\
\hline
\end{tabular}
\end{table}

\begin{figure}[!t]
    \vspace{-0.68cm}
    \centering
    \includegraphics[width=7.8cm]{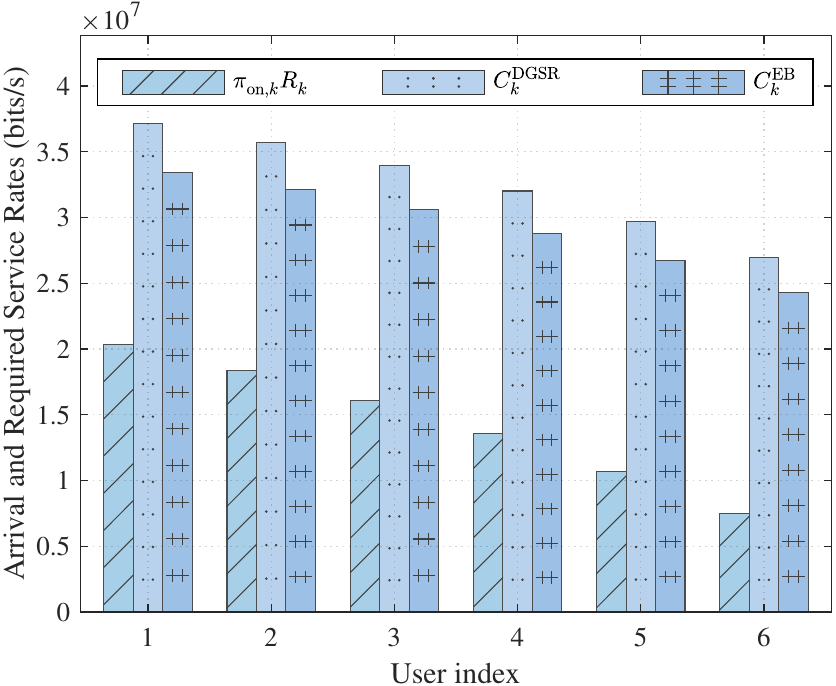}
    \vspace{-0.35cm}
    \caption{Comparison among the average arrival rate, the proposed DGSR, and the EB-based service-rate requirement.}
    \label{fig:DGSR_EB_bar}
\end{figure}

\subsection{Simulation Setup and Benchmark}
\vspace{-0.05cm}
In this section, we present numerical results to evaluate the performance of the proposed  {AO algorithm}. In the simulations, the BS is equipped with $M\!=\!9$ Tx-MAs and serves $K\!=\!6$ users, each equipped with one Rx-MA. The transmission process is organized into consecutive blocks, each containing $N_{0}\!=\!500$ slots with slot duration $\tau\!=\!0.125$ ms. User packet arrivals follow the MMOO model described in Section~II-C. The distance between user $k$ and the BS, denoted by $d_k^{\mathrm{BU}}$, is randomly generated within $[100,200]$ m. For all users, the Tx-side and Rx-side path numbers are identical, i.e., $L_k^{\rm s}\!=\!L_k^{\rm r}\!=\!L\!=\!6$, $\forall k\!\in\!\mathcal K$. The elevation and azimuth AoDs/AoAs are independently generated from the uniform distribution over $[0,\pi]$. The path-response matrix $\boldsymbol{\Sigma}_k$ is diagonal, with diagonal elements following $\mathcal{CN}(0,\Gamma_k^2/L)$, where $\Gamma_k^2=g_0 d_k^{-\delta_{\mathrm{loss}}}$, $g_0=-42$ dB, and $\delta_{\mathrm{loss}}=2.8$. Other key simulation parameters are specified in the  Table~\ref{tab:simulation_parameters}. 

To manifest the superiority of the proposed DGSR, we also show the performance of the existing EB-based scheme as a benchmark, which is also constructed based on the same MMOO arrival model. According to EB theory \cite{elwalid1993effective}, the EB-based service rate is given by
\begin{equation}
\vspace{-0.1cm}
\hat C_k^{\mathrm{EB}}
\triangleq
\frac{1}{\varsigma_k}
\ln sp\!\left(
\mathbf{P}_k
\mathrm{diag}
\left(
1,e^{\varsigma_k R_k}
\right)
\right), 
\label{eq:EB_service_rate}
\end{equation}
where $\varsigma_k>0$ denotes the statistical QoS exponent. Under the conventional large-deviations approximation, the QoS exponent is determined as
\begin{equation}
\vspace{-0.1cm}
\varsigma_k
\triangleq
\inf\!\left\{
\varsigma>0\!:\!
\left[
sp\!\left(
\mathbf{P}_k
\mathrm{diag}
\left(
1,e^{\varsigma R_k}
\right)
\right)
\right]^{-D_k^{\max}}
\! \leq \!
\varepsilon_k
\right\}.
\label{eq:EB_exponent}
\end{equation}
Unlike the proposed DGSR, the EB-based requirement does not account for the backlog accumulation and service interruption induced by MA movement.

\begin{figure*}[!t]
    \vspace{-0.65cm}
    \centering
    \includegraphics[width=\textwidth]{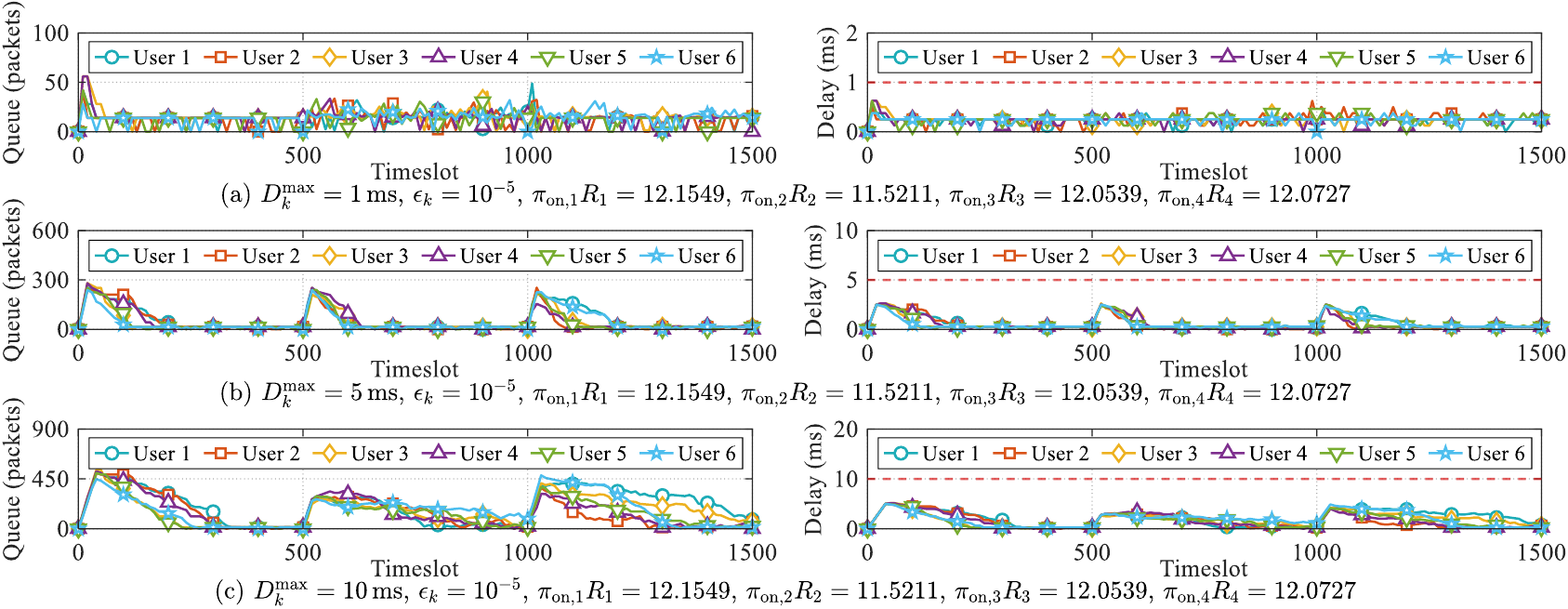}
    \vspace{-0.75cm}
    \caption{Queue backlog and actual delay evolution by the proposed {AO algorithm}.}
    \vspace{-0.3cm}
    \label{fig:DORA_queue_delay}
\end{figure*}

\begin{figure*}[!t]
    \centering
    \IfFileExists{figure7.pdf}{%
        \includegraphics[width=17.75cm]{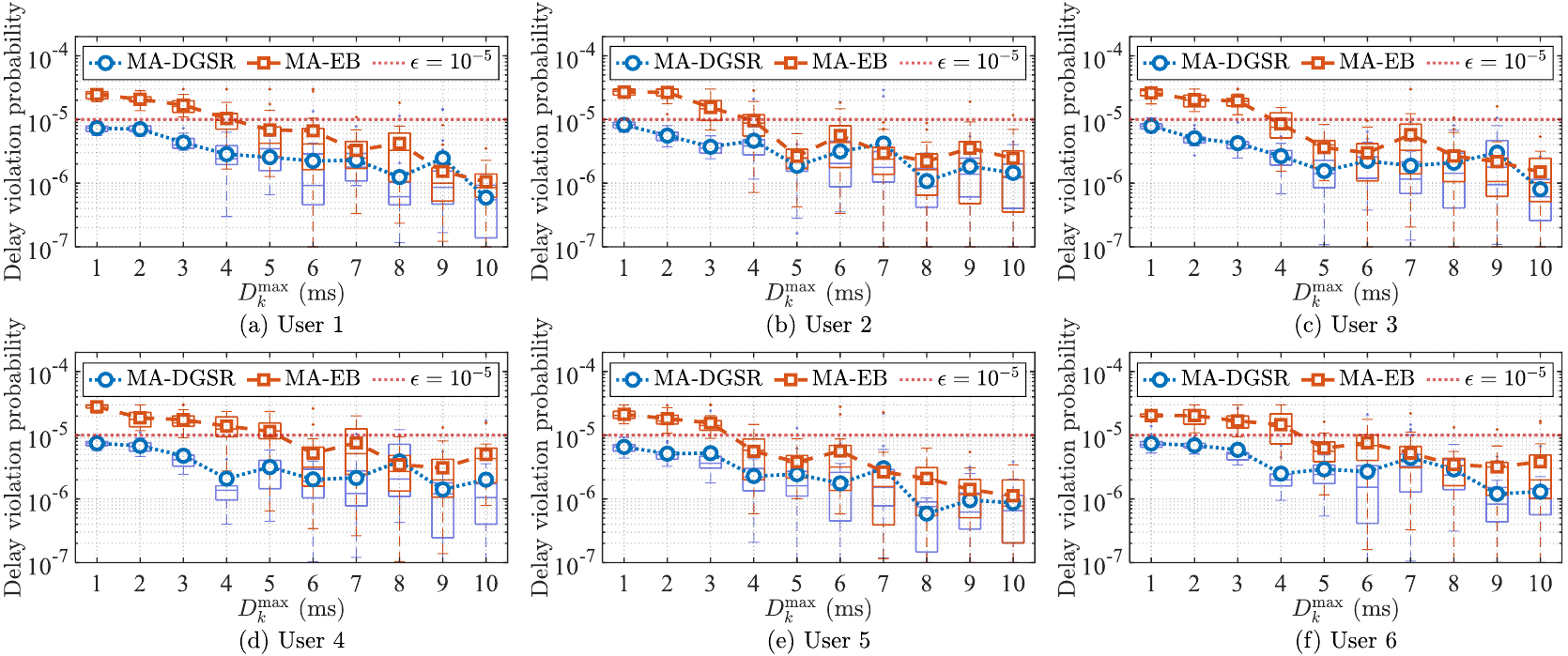}
        \vspace{-0.3cm}
        \caption{Empirical delay-outage probabilities of the MA-DGSR and MA-EB schemes.}
    }{%
        \includegraphics[width=17.75cm]{figure7.pdf}
        \vspace{-0.3cm}
        \caption{Empirical delay-outage probabilities of MA-DGSR and MA-EB.}
    }
    \label{fig:DORA_delay_violation}
\end{figure*}

\subsection{{Effects} of DGSR}
\vspace{-0.1cm}
We first evaluate the DGSR obtained from the backlog martingale analysis. In Fig.~\ref{fig:DGSR_characterization}, we plot the proposed DGSR $C_k^{\mathrm{DGSR}}$ and the corresponding martingale exponent $\zeta_k^\star$ for all users. It is observed from Figs.~\ref{fig:DGSR_characterization}(a) and \ref{fig:DGSR_characterization}(b) that, as $D_k^{\max}$ increases, both $C_{k}^{\rm DGSR}$ and $\zeta_k^\star$ decrease for all users, since the enlarged delay threshold relaxes the required service rate as well as the exponential decay requirement imposed by the delay-outage constraint. 

Conversely, it is observed from Figs.~\ref{fig:DGSR_characterization}(c) and \ref{fig:DGSR_characterization}(d) that decreasing $\varepsilon_k$ increases both $C_{k}^{\rm DGSR}$ and $\zeta_k^\star$. This indicates that a more stringent delay-outage constraint requires a higher service capability and a faster decay of the delay tail. It is also observed from Figs.~\ref{fig:DGSR_characterization}(e) and \ref{fig:DGSR_characterization}(f) that increasing $N_{\rm mov}(\xi)$ also raises the required service capability and the decay of delay tail, since a longer MA movement interval reduces the available DT duration while arrivals continue to accumulate during the MA movement. These results demonstrate that the DGSR jointly captures the impacts of the available delay budget, the delay-outage constraint, and the movement-induced service debt.

Based on the EB benchmark defined in \eqref{eq:EB_service_rate}, Fig.~\ref{fig:DGSR_EB_bar} compares the average arrival rate $\pi_{k,\mathrm{on}}R_k$, the DGSR $C_k^{\mathrm{DGSR}}$, and the EB-based service rate $C_k^{\mathrm{EB}}$. It is observed that both service rates exceed the average arrival rate because of traffic burstiness. The proposed DGSR is consistently higher than $C_k^{\mathrm{EB}}$ because it accounts for the service debt accumulated within each block, the movement-induced delay offset, and the movement-induced zero-service interval. In contrast, EB mainly captures steady-state arrival burstiness. Therefore, the DGSR provides a more appropriate service-rate requirement than EB.

\subsection{Performance of Proposed AO Algorithm}

\begin{figure*}[!t]
    \centering
    \includegraphics[width=18cm]{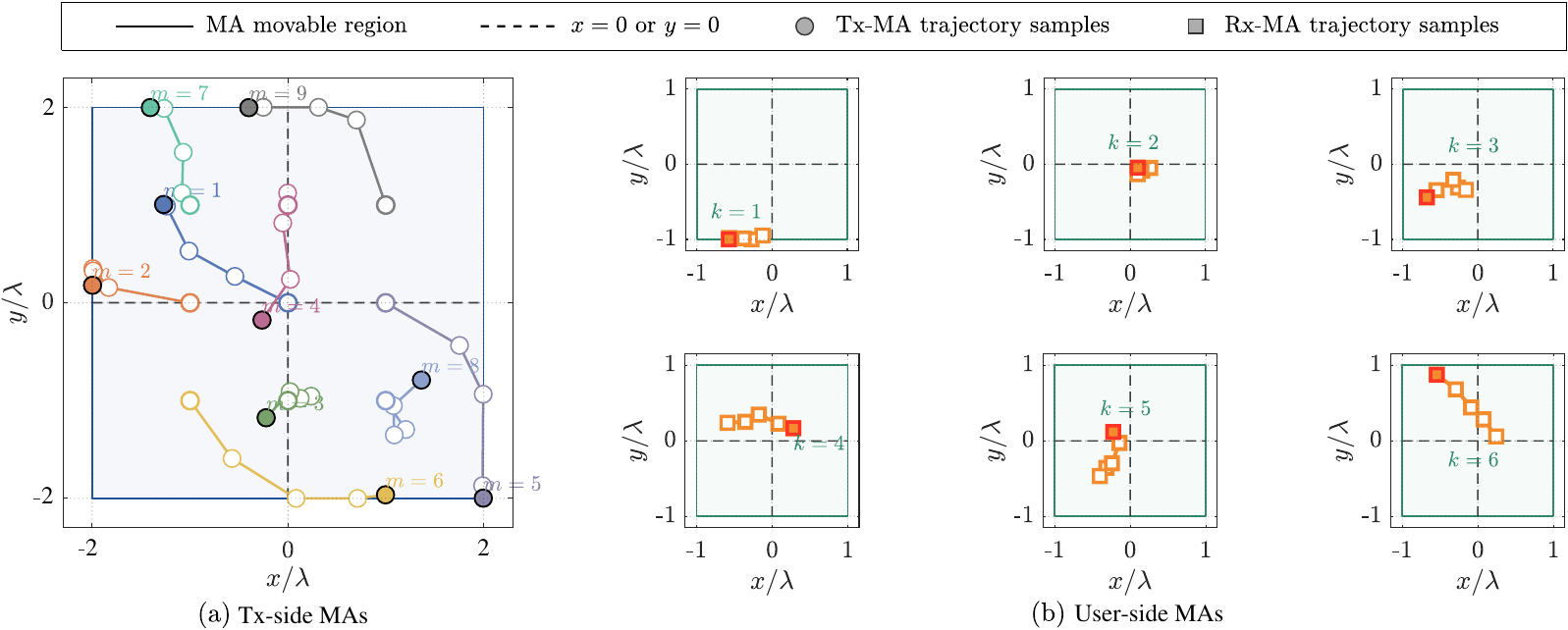}
    \vspace{-0.3cm}
    \caption{MA trajectory samples at representative transmission blocks {by the proposed AO algorithm}.}
    \vspace{-0.45cm}
    \label{fig:DORA_MA_trajectory}
\end{figure*}

We next evaluate the effectiveness of the proposed AO algorithm subject to the DGSR constraint. Fig.~\ref{fig:DORA_queue_delay} presents the queue backlog and actual delay evolution achieved by the proposed algorithm, where the first three transmission blocks, corresponding to 1500 slots, are shown for clarity. It is observed that the queue backlog increases during the ON periods of MMOO arrivals and the MA movement intervals, since data service is temporarily suspended during MA movement. Once the DT phase starts, the backlog is gradually reduced subject to the DGSR constraint. As $D_k^{\max}$ becomes larger, the system can operate at a lower service rate, allowing the queue to remain nonempty for a longer time compared with the case with a smaller $D_k^{\max}$. Nevertheless, the queues remain stable and the actual delays are effectively controlled, which verifies  the effectiveness of the proposed AO algorithm under bursty arrivals and dynamic MA movement.

\begin{figure}[!t]
    \centering
    \includegraphics[width=8cm]{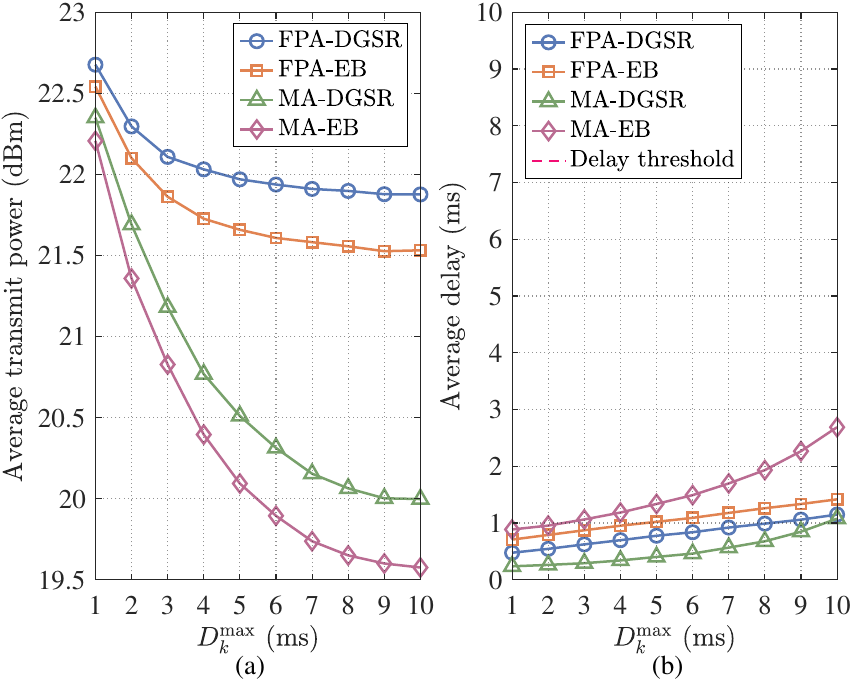}
    \vspace{-0.25cm}
    \caption{Comparison of FPA- and MA-based designs with DGSR- and EB-based service-rate constraints: (a) average transmit power versus $D_k^{\max}$ and (b) average actual delay versus $D_k^{\max}$.}
    \vspace{0.05cm}
    \label{fig:MA_FPA_EB_DGSR}
\end{figure}

The empirical delay-outage probabilities of all users under MA-DGSR and MA-EB are shown in Fig.~\ref{fig:DORA_delay_violation}, obtained from $N_{\rm MC}=50$ independent Monte Carlo realizations, each spanning $|\boldsymbol{\Xi}|=2\times10^4$ transmission blocks with $N_0=500$ slots per block. The red dotted line represents the prescribed {delay-outage} probability $\varepsilon_k=10^{-5}$. It is observed that the delay-outage probabilities of MA-DGSR remain below the target threshold for all users. This suggests that the deterministic DGSR constraint derived from the backlog martingale can satisfy the probabilistic delay-outage constraint. In contrast, MA-EB exceeds the target threshold when $D_k^{\max}$ is small, since its service rate requirement does not explicitly account for movement-induced service interruptions and the resulting backlog accumulation, potentially underestimating the service rate required under tight delay budgets.

Fig.~\ref{fig:DORA_MA_trajectory} shows the Tx- and Rx-MA trajectories sampled at five transmission blocks, i.e., $\xi=0,|\boldsymbol{\Xi}|/4,\ldots,|\boldsymbol{\Xi}|$, in Figs.~\ref{fig:DORA_MA_trajectory}(a) and~\ref{fig:DORA_MA_trajectory}(b), respectively. In Fig.~\ref{fig:DORA_MA_trajectory}(a), the trajectories of different MAs are marked in different colors. It is observed that the sampled MA positions remain within their respective movable regions, while the Tx-MAs satisfy the minimum spacing constraint. The positional changes across blocks demonstrate that the proposed {AO algorithm} adapts the MA configuration to the block-varying channel and queue states while accounting for movement delay.

Finally, Figs.~\ref{fig:MA_FPA_EB_DGSR}(a) and~\ref{fig:MA_FPA_EB_DGSR}(b) compare the average transmit power and average delay of all users under FPA-DGSR, FPA-EB, MA-DGSR, and MA-EB. It is observed from Fig.~\ref{fig:MA_FPA_EB_DGSR}(a) that the MA-based designs require less transmit power than their FPA counterparts. Furthermore, the EB-based designs further reduce the transmit power due to their lower service rate requirements, but this reduction comes at the violation of the delay-outage constraints.  Fig.~\ref{fig:MA_FPA_EB_DGSR}(b) further shows that MA-EB also incurs a larger average delay than MA-DGSR. This demonstrates that the proposed MA-DGSR achieves a more favorable power-delay tradeoff by explicitly accounting for the movement-induced service debt.

\section{Conclusion}

In this paper, we proposed a cross-layer optimization framework for an MA-enhanced MU-MIMO communication system. A two-timescale transmission protocol and a cross-layer queueing model were developed to capture MA movement delay, packet-level queue evolution, and movement-induced service interruption under Markov-modulated bursty arrivals. Based on this model, a long-term transmit-power minimization problem was formulated to jointly optimize the MAs' destination positions and transmit beamforming under delay-outage and practical movement constraints. To handle the probabilistic delay-outage constraints, we constructed a backlog martingale and derived an analytical sufficient condition in terms of the DGSR. An AO algorithm that combines a short-horizon approximation and SCA was also proposed to solve the resulting problem efficiently. Numerical results showed that the proposed algorithm satisfies the probabilistic delay-outage requirements and dramatically outperforms conventional FPA systems. Compared with the EB-based scheme, the proposed framework further mitigates the delay degradation caused by ignoring movement-induced service debt. These results were obtained under the assumptions of block-fading channels, MMOO traffic, and zero service during MA movement. This paper can be extended in several directions in future work. For example, we have assumed that no data transmission occurs during antenna movement and that the time overhead for channel estimation is ignored. It would be interesting to jointly optimize MA movement and communication strategies in a cross-layer framework.

\bibliographystyle{IEEEtran}
\bibliography{reference}

\end{document}